%% file: hash.tex
\def\short{0}  
\def\draft{0}  

\ifnum\short=1
\documentclass[oribibl, envcountsame, envcountsect]{llncs}
\fi
\ifnum\short=0
\documentclass[11pt]{article}
\fi

\usepackage{amsfonts,latexsym,graphics}
\usepackage{epsfig}
\usepackage{amssymb}
\usepackage{amsmath}
\usepackage{amsfonts}

\input{macros}

\newcommand{\cp}{\mathrm{cp}}
\newcommand{\supp}{\mathrm{supp}}

\newcommand{\vecX}{\mathbf{X}}
\newcommand{\vecY}{\mathbf{Y}}
\newcommand{\vecZ}{\mathbf{Z}}

\newcommand{\vecx}{\mathbf{x}}
\newcommand{\vecy}{\mathbf{y}}
\newcommand{\vecz}{\mathbf{z}}

\newcommand{\Hfam}{\mathcal{H}}

\newcommand{\Hmin}{\mathrm{H}_\infty}
\newcommand{\HRen}{\mathrm{H}_2}

\newcommand{\wt}{\mathrm{wt}}

\ifnum\draft=1
\newcommand{\Snote}[1]{{[\bf Salil's Note: #1]}}
\newcommand{\Cnote}[1]{{[\bf Chung's Note: #1]}}
\else
\newcommand{\Snote}[1]{{}}
\newcommand{\Cnote}[1]{{}}
\fi

\ifnum\short=0
\title{Tight Bounds for Hashing Block Sources\thanks{An extended abstract of this paper will appear in {\em
RANDOM `08} \cite{ChungVa08}.}}
\else
\title{Tight Bounds for Hashing Block Sources\thanks{A full version of this paper can be found on \cite{ChungVa08TR}.}}
\fi

\date{\today}

\ifnum\short=1
\author{Kai-Min Chung\thanks{Work done when visiting U.C. Berkeley,
supported by US-Israel BSF grant 2006060 and NSF grant CNS-0430336.}
 \and Salil Vadhan\thanks{Work done when visiting U.C. Berkeley, supported
 by the Miller Institute for Basic Research in Science, a Guggenheim 
Fellowship, US-Israel BSF grant 2006060, and ONR grant N00014-04-1-0478.}}

\institute{School of Engineering \& Applied Sciences \\ Harvard University \\ Cambridge, MA\\
\email{\{kmchung,salil\}@eecs.harvard.edu}}
\fi

\ifnum\short=0
\author{Kai-Min Chung\thanks{Work done when visiting U.C. Berkeley,
supported by US-Israel BSF grant 2002246 and NSF grant CNS-0430336.}\\ Harvard University
 \and Salil Vadhan\thanks{Work done when visiting U.C. Berkeley, supported
 by the Miller Institute for Basic Research in Science, a Guggenheim 
Fellowship, US-Israel BSF grant 2006060, and ONR grant N00014-04-1-0478.}\\ Harvard 
University }
\fi

\begin{document}


\maketitle

\begin{abstract}
It is known that if a 2-universal hash function $H$ is applied to elements of a {\em block source} $(X_1,\ldots,X_T)$, where each item $X_i$ has enough min-entropy conditioned on the previous items, then the output distribution $(H,H(X_1),\ldots,H(X_T))$ will be ``close'' to the uniform distribution. We provide improved bounds on how much min-entropy per item is required for this to hold, both when we ask that the output be close to uniform in statistical distance and when we only ask that it be statistically close to a distribution with small collision probability.  In both cases, we reduce the dependence of the min-entropy on the number $T$ of items from $2\log T$ in previous work to $\log T$, which we show to be optimal.  This leads to corresponding improvements to the recent results of Mitzenmacher and Vadhan (SODA `08) on the analysis of hashing-based algorithms and data structures when the data items come from a block source.
\end{abstract}



\section{Introduction}


A {\em block source} is a sequence of items $\vecX =
(X_1,\dots,X_T)$ in which each item has at least some $k$ bits of
``entropy" conditioned on the previous ones~\cite{ChorG88}. Previous
works~\cite{ChorG88,Zuckerman96,MitzenmacherV08} have analyzed what
happens when one applies a 2-universal hash function to each item in
such a sequence, establishing results of the following form:

\begin{quote}
\textbf{Block-Source Hashing Theorems (informal):} \textit{If
$(X_1,\ldots,X_T)$ is a block source with $k$ bits of ``entropy''
per item and $H$ is a random hash function from a 2-universal family
mapping to $m\ll k$ bits, then $(H(X_1),\ldots,H(X_T))$ is ``close''
to the uniform distribution.}
\end{quote}

In this paper, we prove new results of this form, achieving improved
(in some cases, optimal) bounds on how much entropy $k$ per item is
needed to ensure that the output is close to uniform, as a function
of the other parameters (the output length $m$ of the hash
functions, the number $T$ of items, and the ``distance'' from the
uniform distribution).  But first we discuss the two applications
that have motivated the study of Block-Source Hashing Theorems.

\subsection{Applications of Block-Source Hashing}

\paragraph{Randomness Extractors.} A {\em randomness extractor} is an
algorithm that extracts almost-uniform bits from a source of biased
and correlated bits, using a short {\em seed} of truly random bits as a
catalyst~\cite{NisanZu96}.  Extractors have many applications in
theoretical computer science and have played a central role in the
theory of pseudorandomness.  (See the surveys
\cite{NisanT99,Shaltiel04,Vadhan07-SIGACT}.)  Block-source Hashing Theorems
immediately yield methods for extracting randomness from block
sources, where the seed is used to specify a universal hash
function.  The gain over hashing the entire $T$-tuple at once
is that the blocks may be much shorter than the entire
sequence, and thus a much shorter seed is required to specify the
universal hash function.  Moreover, many subsequent constructions
of extractors for general sources (without the block structure) work
by first converting the source into a block source and performing
block-source hashing.

\paragraph{Analysis of Hashing-Based Algorithms.}
The idea of hashing has been widely applied in designing algorithms
and data structures, including hash tables~\cite{Knuth98}, Bloom
filters~\cite{BroderM03}, summary algorithms for data
streams~\cite{Muthukrishnan03}, etc.  Given a stream of data items
$(x_1,\dots,x_T)$, we first hash the items into \\
$(H(x_1),\dots,H(x_T))$, and carry out a computation using the
hashed values. In the literature, the analysis of a hashing algorithm
is typically a worst-case analysis on the input data items, and the
best results are often obtained by unrealistically modelling the
hash function as a truly random function mapping the items
to uniform and independent
$m$-bit strings.  On the other hand, for realistic, efficiently
computable hash functions (eg., $2$-universal or $O(1)$-wise
independent hash functions), the provable performance is sometimes
significantly worse.  However, such gaps seem to not show up in
practice, and even standard 2-universal hash functions empirically seem to
match
the performance of truly random hash functions.  To explain this
phenomenon, Mitzenmacher and Vadhan \cite{MitzenmacherV08} have
suggested that the discrepancy is due to worst-case analysis,
and propose to instead model the input items as coming from a block source.
Then Block-Source Hashing Theorems imply that the performance of
universal hash functions is close to that of truly random hash
functions, provided that each item has enough bits of entropy.
%

\subsection{How Much Entropy is Required?}

A natural question about Block-Source Hashing Theorems is: how large does the ``entropy'' $k$ per item need to be to ensure a certain amount of ``closeness" to uniform (where both the entropy and closeness can be measured in various ways).  This also has practical significance for the latter motivation regarding hashing-based algorithms, as it corresponds to the amount of entropy we need to assume in data items.  In \cite{MitzenmacherV08}, they provide bounds on the entropy required for two measures of closeness, and use these as basic tools to bound the required entropy in various applications. The requirement is usually some small constant multiple of $\log T$, where $T$ is the number of items in the source, which can be on the borderline between a reasonable and unreasonable assumption about real-life data. Therefore, it is interesting to pin down the optimal answers to these questions. In what follows, we first summarize the previous results, and then discuss our improved analysis and corresponding lower bounds.

A standard way to measure the distance of the output from the
uniform distribution is by \emph{statistical distance}.\footnote{The
\emph{statistical distance} of two random variables $X$ and $Y$ is
$\Delta(X,Y) = \max_T | \Pr[X \in T] - \Pr[Y \in T]|$, where $T$ ranges over all possible events.}
In the randomness extractor literature, classic results
\cite{ChorG88,ImpagliazzoLL89,Zuckerman96} show that using
2-universal hash functions, $k = m + 2 \log(T/\eps) + O(1)$ bits of
min-entropy (or even Renyi entropy)\footnote{The {\em min-entropy} of a 
random variable $X$
is $\Hmin(X)=\min_x \log(1/\Pr[X=x])$.  All of the results mentioned
actually hold for the less stringent measure of {\em Renyi entropy}
$\HRen(X)= \log(1/\Exp_{x\getsr X}[\Pr[X=x]])$.} per item is sufficient for
the output distribution to be $\eps$-close to uniform in statistical
distance. Sometimes a less
stringent closeness requirement is sufficient, where we only require
that the output distribution is $\eps$-close to a distribution having
``small" \emph{collision probability}\footnote{The \emph{collision
probability} of a random variable $X$ is $\sum_x \Pr[X=x]^2$.  By
``small collision probability,'' we mean that the collision
probability is within a constant factor of the collision probability
of uniform distribution.}. A result of \cite{MitzenmacherV08} shows
that $k = m + 2\log T + \log (1/\eps) + O(1)$ suffices to achieve
this requirement. Using $4$-wise independent hash functions,
\cite{MitzenmacherV08} further reduce the required entropy to $k=
\max\{m + \log T, 1/2 (m + 3\log T + \log (1/\eps))\} + O(1)$.

\begin{table}[t]
\begin{center} {
\begin{tabular}{|c|c|c|}
 \hline
 Setting & Previous Results & Our Results \\ \hline
 2-universal hashing & $m + 2 \log T + 2\log (1/\eps)$
& $m +  \log T + 2 \log (1/\eps)$ \\
 $\eps$-close to uniform & \cite{ChorG88,ImpagliazzoLL89,Zuckerman96} & \\ 
\hline
 2-universal hashing & $m + 2\log T + \log
 (1/\eps)$~\cite{MitzenmacherV08} & $m +  \log T + \log (1/\eps)$ \\
 $\eps$-close to small cp. & & \\ \hline
 4-wise indep. hashing & \multicolumn{1}{l|}{$\max\{m + \log T,$} &  
\multicolumn{1}{l|}{$\max\{m + \log T,$}\\
 $\eps$-close to small cp. &
 $1/2 (m + 3\log T + \log 1/\eps)\}$~\cite{MitzenmacherV08}& $1/2(m + 2\log 
T +  \log (1/\eps)\}$\\
\hline
\end{tabular} }
\end{center}
\caption{Our Results: Each entry denotes the min-entropy (actually,
Renyi entropy) required per item when hashing a block source of $T$
items to $m$-bit strings to ensure that the output has statistical distance 
at most
$\eps$ from uniform (or from having collision probability within a constant 
factor of
uniform).  Additive constants are omitted for readability.}
\label{tbl:analysis}
\end{table}

\paragraph{Our Results.} We reduce the entropy required in the previous results, as summarized in Table \ref{tbl:analysis}. Roughly speaking, we save an additive $\log T$ bits of min-entropy (or Renyi entropy) for all cases. We show that using universal hash functions, $k = m + \log T + 2\log 1/\eps + O(1)$ bits per item is sufficient for the output to be $\eps$-close to uniform, and $k = m + \log (T/\eps) + O(1)$ is enough for the output to be $\eps$-close to having small collision probability. Using $4$-wise independent hash functions, the entropy $k$ further reduces to $ \max\{m + \log T, 1/2 (m + 2\log T + \log 1/\eps)\} + O(1)$. The results hold even if we consider the joint distribution $(H,H(X_1),\dots,H(X_T))$ (corresponding to ``strong extractors'' in the literature on randomness extractors). Substituting our improved bounds in the analysis of hashing-based algorithms from \cite{MitzenmacherV08}, we obtain similar reductions in the min-entropy required for every application with 2-universal hashing. With 4-wise independent hashing, we obtain a slight improvement for Linear Probing, and for the other applications, we show that the previous bounds can already be achieved with 2-universal hashing.  The results are summarized in Table \ref{tbl:app}.  


Although the $\log T$ improvement seems small, we remark that it could be significant for practical settings of parameter. For example, suppose we want to hash $64$ thousand internet traffic flows, so $\log T \approx 16$.  Each flow is specified by the 32-bit IP addresses and 16-bit port numbers for the source and destination plus the 8-bit transport protocol, for a total of 104 bits. There is a noticeable difference between assuming that each flow contains $3\log T\approx 48$ vs. $4\log T\approx 64$ bits of entropy as they are only 104 bits long, and are very structured.

%

We also prove corresponding lower bounds showing that our upper
bounds are almost tight.
Specifically, we show that when the data items have not enough
entropy, then the joint distribution \ifnum\short=1\\\fi 
$(H,H(X_1),\dots,H(X_T))$ can
be ``far" from uniform. More precisely, we show that if $k = m +
\log T + 2\log 1/\eps - O(1)$, then there exists a block source
$(X_1,\dots,X_T)$ with $k$ bits of min-entropy per item such that
the distribution $(H,H(X_1),\dots,H(X_T))$ is $\eps$-far from
uniform in statistical distance (for $H$ coming from any hash
family). This matches our upper bound up to an additive constant.
Similarly, we show that if $k = m + \log T - O(1)$, then there
exists a block source $(X_1,\dots,X_T)$ with $k$ bits of min-entropy
per item such that the distribution $(H,H(X_1),\dots,H(X_T))$ is
$0.99$-far from having small collision probability (for $H$ coming
from any hash family).  This matches our upper bound up to an
additive constant in case the statistical distance parameter $\eps$
is constant; we also exhibit a specific 2-universal family for which
the $\log(1/\eps)$ in our upper bound is nearly tight --- it cannot
be reduced below $\log(1/\eps)-\log\log(1/\eps)$.  Finally, we also
extend all of our lower bounds to the case that we only consider
distribution of hashed values $(H(X_1),\ldots,H(X_T))$, rather than
their joint distribution with $Y$.  For this case, the lower bounds are
necessarily reduced by a term that depends on the size of the hash family.
(For standard constructions of universal hash functions, this
amounts to $\log n$ bits of entropy, where $n$ is the bit-length of
an individual item.)

\begin{table}[t]
\begin{center} {
\begin{tabular}{|l|c|c|}
 \hline
Type of Hash Family & Previous Results~\cite{MitzenmacherV08} & Our
Results\\ \hline \hline \multicolumn{3}{|c|}{Linear Probing} \\
\hline
2-universal hashing & $4\log T$ & $3\log T$ \\
4-wise independence & $2.5\log T$ & $2\log T$\\ \hline \hline
\multicolumn{3}{|c|}{Balanced Allocations with $d$ Choices} \\
\hline
2-universal hashing & $(d+2)\log T$ & $(d+1)\log T$ \\
4-wise independence & $(d+1)\log T$ & --- \\ \hline \hline
\multicolumn{3}{|c|}{Bloom Filters} \\ \hline
2-universal hashing & $4\log T$ & $3\log T$ \\
4-wise independence & $3\log T$ & ---\\ \hline
\end{tabular} }
\end{center}
\caption{Applications: Each entry denotes the min-entropy (actually, Renyi entropy) required per item to ensure
that the performance of the given application is ``close'' to the performance when using truly random hash functions.  In all cases, the bounds omit additive terms that depend on how close a performance is desired, and we restrict to the (standard) case that the size of the hash table is linear in the number of items being hashed. That is, $m = \log T + O(1)$.}  \label{tbl:app}
\end{table}

\paragraph{Techniques.}
At a high level, all of the previous analyses for hashing block sources were loose due to summing error probabilities over the $T$ blocks.  Our improvements come from avoiding this linear blow-up by choosing more refined measures of error. For example, when we want the output to have small statistical distance from uniform, the classic Leftover Hash Lemma~\cite{ImpagliazzoLL89} says that min-entropy $k=m+2\log(1/\eps_0)$ suffices for a single hashed block to be $\eps_0$-close to uniform, and then a ``hybrid argument'' implies that the joint distribution of $T$ hashed blocks is $T\eps_0$-close to uniform~\cite{Zuckerman96}.  Setting $\eps_0=\eps/T$, this leads to a min-entropy requirement of $k=m+2\log(1/\eps)+2\log T$ per block.  We obtain a better bound, reducing $2\log T$ to $\log T$, by using {\em Hellinger distance} to analyze the error accumulation over 
blocks, and only passing to statistical distance at the end.

For the case where we only want the output to be close to having small collision probability, the previous
analysis of \cite{MitzenmacherV08} worked by first showing that the expected  collision probability of each hashed block $h(X_i)$ is ``small'' even conditioned on previous blocks, then using Markov's Inequality to deduce that each hashed block has small collision probability except with some probability $\eps_0$, and finally doing a union bound to deduce that all hashed blocks have small collision probability except with probability  $T\eps_0$.  We avoid the union bound by working with more refined notions of ``conditional collision probability,'' which enable us to apply Markov's Inequality on the entire sequence rather than on each block individually.

The starting point for our negative results is the tight lower bound
for randomness extractors due to Radhakrishnan and
Ta-Shma~\cite{RadhakrishnanT00}.  Their methods show that if the
min-entropy parameter $k$ is not large enough, then for any hash
family, there exists a (single-block) source $X$ such that $h(X)$ is
``far'' from uniform (in statistical distance) for ``many'' hash
functions $h$.  We then take our block source $(X_1,\ldots,X_T)$ to
consist of $T$ iid copies of $X$, and argue that the statistical
distance from uniform grows sufficiently fast with the number $T$ of
copies taken.  For example, we show that if two distributions have
statistical distance $\eps$, then their $T$-fold products have
statistical distance $\Omega(\min\{1,\sqrt{T}\cdot\eps\})$,
strengthening a previous bound of Reyzin~\cite{Reyzin04}, who proved
a bound of $\Omega(\min\{\eps^{1/3},\sqrt{T}\cdot\eps\})$.

\section{Preliminaries}

\paragraph{Notations.} All logs are based $2$.  We use the
convention that $N = 2^n$, $K = 2^k$, and $M = 2^m$. 
We think of a data item $X$ as a random variable over $[N] =
\{1,\dots,N\}$, which can be viewed as the set of $n$-bit strings. A
hash function $h : [N] \rightarrow [M]$ hashes an item to a $m$-bit
string. A \emph{hash function family} $\Hfam$ is a multiset of hash
functions, and $H$ will usually denote a uniformly random hash
function drawn from $\Hfam$. $U_{[M]}$ denotes the uniform
distribution over $[M]$. Let $\vecX = (X_1,\dots,X_T)$ be a sequence
of data items. We use $X_{<i}$ to denote the first $i-1$ items
$(X_1,\dots,X_{i-1})$. We refer to $X_i$ as an item or a block
interchangeably. Our goal is to study the distribution of hashed
sequence $(H,\vecY) = (H,Y_1,\dots,Y_T) \eqdef
(H,H(X_1),\dots,H(X_T))$.

\paragraph{Hash Families.} The \emph{truly random hash family}
$\Hfam$ is the set of all functions from $[N]$ to $[M]$. A hash
family $\Hfam$ is \emph{$s$-universal} if for every sequence of
distinct elements $x_1,\dots,x_s \in [N]$, $\Pr_H[H(x_1) = \dots =
H(x_s)] \leq 1/M^s$. $\Hfam$ is \emph{$s$-wise} independent if for
every sequence of distinct elements $x_1,\dots,x_s \in [N]$,
$H(x_1),\dots,H(x_s)$ are independent and uniform random variables
over $[M]$.

\paragraph{Block Sources and Collision Probability.} For a random
variable $X$, the collision probability of $X$ is $\cp(X) = \Pr[X =
X'] = \sum_x \Pr[X = x]^2$, where $X'$ is an independent copy of
$X$. The \emph{Renyi entropy} $\HRen(X) = \log (1/\cp(X))$ can be
viewed as a measure of the amount of randomness in $X$ (In
the randomness extractor literature, the entropy is measured by
\emph{min-entropy} $H_{\infty}(X) = \min_{x\in \supp(X)} \log (1/\Pr[X=x])$, 
but using the less stringent measure Renyi entropy makes our results 
stronger since $H_{2}(X) \geq H_{\infty}(X)$.)
For an event $E$, $(X|_E)$ is the
random variable defined by conditioning $X$ on $E$.

\begin{definition}[Block Sources]
A sequence of random variables $(X_1,\dots,X_T)$ over $[N]^T$ is a
\emph{block $K$-source} if for every $i \in [T]$, and every $x_{<i}$
in the support of $X_{<i}$, we have $\cp(X_i| X_{<i} = x_{<i}) \leq
1/K$. That is, each item $X_i$ has at least $k = \log K$ bits of
Renyi entropy even after conditioning on the previous items.
\end{definition}

Let $\vecX = (X_1,\dots,X_T)$ be a sequence of random variables over
$[M]^T$. We are interested in bounding the overall collision
probability $\cp(\vecX)$ by the collision probability of each
blocks. Suppose all $X_i$'s are independent, then $\cp(\vecX) =
\prod_{i=1}^T \cp(X_i)$.  The following lemma
generalizes Lemma 4.2
in \cite{MitzenmacherV08}, which says that if for every $\vecx \in
\vecX$, the average collision probability of every block $X_i$
conditioning on $X_{<i} = x_{<i}$ is small, then the overall
collision probability $\cp(\vecX)$ is also small. In particular, if
$\vecX$ is a block $K$-source, then $\cp(\vecX) \leq 1/K^T$.
\ifnum\short=1
(For the proof of the lemma, please refer to the full version of this paper\cite{ChungVa08TR}.)
\fi 

\begin{lemma} \label{lem:cp_upper_bound}
  Let $\vecX = (X_1,\dots,X_T)$ be a sequence of random variables
  such that for every $\vecx \in \supp(\vecX)$,
  $$\frac{1}{T}\sum_{i=1}^T \cp(X_i|_{X_{<i} = x_{<i}}) \leq \alpha.$$
  Then the overall collision probability satisfies $\cp(\vecX) \leq \alpha^T
$.
\end{lemma}

\ifnum\short=0
\begin{proof}
By Arithmetic Mean-Geometric Mean Inequality, the
inequality in the premise implies
$$\prod_{i=1}^T \cp(X_i|_{X_{<i} = x_{<i}}) \leq \alpha^T.$$
Therefore, it suffices to prove
$$\cp(\vecX) \leq \max_{\vecx\in \supp(\vecX)} \prod_{i=1}^T \cp(X_i|
_{X_{<i} =
  x_{<i}}).$$
 We prove it by induction on $T$.  The base case $T=1$ is trivial.
  Suppose the lemma is true for $T-1$. We have
  \begin{eqnarray*}
    \cp(\vecX) & = & \sum_{x_1} \Pr[X_1 = x_1]^2 \cdot \cp(X_2,\dots,X_T|_{X_1 = x_1}) \\
               & \leq & \left(\sum_{x_1} \Pr[X_1 = x_1]^2\right) \cdot  \max_{x_1} \cp(X_2,\dots,X_T|_{X_1 = x_1}) \\
               & \leq & \cp(X_1) \cdot \max_{x_1} \left( \max_{x_2,\dots,x_T} \prod_{i=2}^T \cp(X_i|_{X_{<i} = x_{<i}}) \right)  \\
               & = & \max_{\vecx} \prod_{i=1}^T \cp(X_i|_{X_{<i} = x_{<i}}),
  \end{eqnarray*}
  as desired.
\end{proof}
\fi

\paragraph{Statistical Distance.} The statistical distance is a
standard way to measure the distance of two distributions. Let $X$
and $Y$ be two random variables. The \emph{statistical distance} of $X$ and
$Y$ is $\Delta(X,Y) = \max_T |\Pr[X \in T] - \Pr[Y \in T]| = (1/2)
\cdot \sum_x |\Pr[X=x] - \Pr[Y=x]|$, where $T$ ranges over all
possible events. When $\Delta(X,Y)\leq \eps$, we say that $X$ is
\emph{$\eps$-close} to $Y$. Similarly, if $\Delta(X,Y)\geq\eps$, then $X$
is \emph{$\eps$-far} from $Y$. The following standard lemma says that if $X$ 
has
small collision probability, then $X$ is close to uniform in
statistical distance.

\begin{lemma} \label{lem:cp_to_stat}
  Let $X$ be a random variable over $[M]$ such that $\cp(X) \leq
  (1+\eps)/M$. Then $\Delta(X,U_{[M]}) \leq \sqrt{\eps}$.
\end{lemma}

%

\paragraph{Conditional Collision Probability.} Let $(X,Y)$ be jointly
distributed random variables. We can define the conditional Renyi
entropy of $X$ conditioning on $Y$ as follows.

\begin{definition} \label{def:cond_cp}
The \emph{conditional collision probability} of $X$ conditioning on $Y$
is $\cp(X|Y) =$ \\ $\E_{y\leftarrow Y}[\cp(X|_{Y=y})]$. The
\emph{conditional Renyi entropy} is $\HRen(X|Y) = \log 1/\cp(X|Y)$.
\end{definition}

The following lemma
says that as in the case of Shannon entropy, conditioning can only decrease the entropy.
\ifnum\short=1
(For the proof of the lemma, please refer to the full version of this paper\cite{ChungVa08TR}.)
\fi 

\begin{lemma} \label{lem:cond_cp}
Let $(X,Y,Z)$ be jointly distributed random variables. We have
$\cp(X) \leq \cp(X|Y) \leq \cp(X|Y,Z)$.
\end{lemma}

\ifnum\short=0
\begin{proof}
  For the first inequality, we have
\begin{eqnarray*}
  \cp(X) & = & \sum_x \Pr[X=x]^2\\
        & = & \sum_{y,y'} \Pr[Y=y] \cdot \Pr[Y=y'] \cdot \left(\sum_x
        \Pr[X=x | Y = y] \cdot \Pr[X=x | Y = y']\right) \\
        & \leq & \sum_{y,y'} \Pr[Y=y] \cdot \Pr[Y=y'] \cdot \left(\sum_x 
\Pr[X = x|Y
        =y]^2\right)^{1/2} \cdot \left(\sum_x \Pr[X = x|Y 
=y']^2\right)^{1/2} \\
        & = & \E_{y\leftarrow Y}\left[\cp(X|Y=y)^{1/2} \right]^2\\
        & \leq &  \cp(X|Y) \\ 
\end{eqnarray*}
For the second inequality, observe that for every $y$ in the support
of $Y$, we have $\cp(X|_{Y=y}) \leq \cp((X|_{Y=y})|(Z|_{Y=y}))$ from
the first inequality. It follows that
\begin{eqnarray*}
\cp(X|Y) & = & \E_{y\leftarrow Y}[\cp(X|_{Y=y})] \\
         & \leq & \E_{y\leftarrow Y}[\cp((X|_{Y=y})|(Z|_{Y=y}))] \\
         & = & \E_{y\leftarrow Y}[ \E_{z\leftarrow (Z|Y=y)}[\cp(X|
_{Y=y,Z=z})] \\
         & = & \cp(X|Y,Z)
\end{eqnarray*}
\end{proof}
\fi

\section{Positive Results: How Much Entropy is Sufficient?}

In this section, we present our positive results, showing that the
distribution of hashed sequence $(H,\vecY) =
(H,H(X_1),\dots,H(X_T))$ is close to uniform when $H$ is a random
hash function from a $2$-universal hash family, and $\vecX =
(X_1,\dots,X_T)$ has sufficient entropy per block. The new
contribution is that we will not need $K = 2^k$ to be as large as in
previous works, and so save the required randomness in the block
source $\vecX = (X_1,\dots,X_T)$.

\subsection{Small Collision Probability Using $2$-universal Hash
Functions}

	Let $H:[N]\rightarrow [M]$ be a random hash function from a $2$-universal family $\Hfam$. We first study the conditions under which $(H,\vecY) = (H,H(X_1),\dots,H(X_T))$ is $\eps$-close to having collision probability $O(1/(|\Hfam|\cdot M^T))$. This requirement is less stringent than $(H,\vecY)$ being $\eps$-close to uniform in statistical distance, and so requires less bits of entropy. Mitzenmacher and Vadhan~\cite{MitzenmacherV08} show that this guarantee	suffices for some hashing applications. They show that $K \geq MT^2/\eps$ is enough to satisfy the requirement. We save a factor of $T$, and show that in fact, $K \geq MT/\eps$, is sufficient. (Taking logs yields the first entry in Table \ref{tbl:analysis}, i.e. it suffices to have Renyi entropy $k=m+\log T+\log(1/\eps)$ per block.) Formally, we prove the following theorem.

\begin{theorem} \label{thm:2-univ_cp}
  Let $H:[N]\rightarrow[M]$ be a random hash function from a
  $2$-universal family $\Hfam$. Let $\vecX = (X_1,\dots,X_T)$ be a
  block $K$-source over $[N]^T$. For every $\eps > 0$, the hashed
  sequence
  $(H,\vecY) = $ \\ $(H,H(X_1),\dots,H(X_T))$ is $\eps$-close to a
  distribution $(H,\vecZ) = (H,Z_1,\dots,Z_T)$ such that
  $$\cp(H,\vecZ) \leq \frac{1}{|\Hfam|\cdot M^T}\left(
  1+\frac{M}{K\eps}\right)^T.$$
  In particular, if $K \geq MT/\eps$, then $(H,\vecZ)$ has collision
  probability at most $(1+2MT/K\eps)/(|\Hfam|\cdot M^T$).
\end{theorem}

To analyze the distribution of the hashed sequence $(H,\vecY)$, the
starting point is the following version of the
Leftover Hash Lemma~\cite{BennettBrRo85,ImpagliazzoLL89}, which says that 
when we
hash a random variable $X$ with enough entropy using a $2$-universal
hash function $H$, the conditional collision probability of $H(X)$
conditioning on $H$ is small.

\begin{lemma}[The Leftover Hash Lemma] \label{lem:LHL}
Let $H:[N]\rightarrow [M]$ be a random hash function from a
$2$-universal family $\Hfam$. Let $X$ be a random variable over
$[N]$ with $\cp(X) \leq 1/K$. We have $\cp(H(X)|H) \leq 1/M + 1/K$.
\end{lemma}

We now sketch how the hashed block source $\vecY=(Y_1,\ldots,Y_T)=
(H(X_1),\ldots,H(X_T))$ is analyzed in \cite{MitzenmacherV08}, and how
we improve the analysis.
The following natural approach is taken in
\cite{MitzenmacherV08}. Since the data $\vecX$ is a block
$K$-source, the Leftover Hash Lemma tells us that for every block
$i\in [T]$, if we condition on the previous blocks $X_{<i}=x_{<i}$, then the
hashed value $(Y_i|_{X_{<i}=x_{<i}})$ has small conditional
collision probability, i.e. $\cp((Y_i|_{X_{<i} = x_{<i}})|H) \leq 1/M +
1/K$. This is equivalent to saying that the average collision
probability of $(Y_i|_{X_{<i} = x_{<i}})$ over the choice of the
hash function $H$ is small, i.e.,
$$\E_{h\leftarrow H}[ \cp(h(X_i)|_{X_{<i} = x_{<i}})] = \cp((Y_i|_{X_{<i} =
x_{<i}})|H) \leq \frac{1}{M} + \frac{1}{K}.$$ We can then use a
Markov argument to say that for every block, with probability at
least $1-\eps/T$ over $h\leftarrow H$, the collision probability is
at most $1/M + T/(K\eps)$. We can then take a union bound to say
that for every $\vecx \in \supp(\vecX)$, at least
$(1-\eps)$-fraction of hash functions $h$ are good in the sense that
$\cp(h(X_i)|_{X_{<i} = x_{<i}})$ is small for all blocks $i =
1,\dots,T$. \cite{MitzenmacherV08} shows that if this condition is
true for every $(h,\vecx)\in \supp(H,\vecX)$, then $\vecY$ is a
block $(1/M+T/(K\eps))$-source, and thus the overall collision
probability is at most $(1+MT/K\eps)^T/M^T$. \cite{MitzenmacherV08}
also shows how to modify an $\eps$-fraction of the distribution to
fix the bad hash functions, and thus complete the analysis.

The problem of the above analysis is that taking a Markov argument
for each block, and then taking a union bound incurs a loss of
factor $T$. To avoid this, we want to apply Markov argument only
once to the whole sequence. For example, a natural thing to try is
to sum over blocks to get
$$\E_{h\leftarrow H}\left[ \frac{1}{T}\sum_{i=1}^T \cp(h(X_i)|_{X_{<i} = 
x_{<i}})\right] = \frac{1}{T} \sum_{i=1}^T \cp((Y_i|_{X_{<i} =
x_{<i}})|H) \leq \frac{1}{M} + \frac{1}{K},$$ and use a Markov argument
to deduce that for every $\vecx \in \supp(\vecX)$, with probability
$1-\eps$ over $h\leftarrow H$, the average collision probability per block
satisfies $$\frac{1}{T}\cdot\sum_{i=1}^T \cp(h(X_i)|_{X_{<i} =
x_{<i}}) \leq \frac{1}{M} + \frac{1}{K\eps}.$$  We need to bound the
collision probability of $\vecY$ using this information. We may want
to apply Lemma \ref{lem:cp_upper_bound}, but it requires the
information on $(1/T)\sum_{i=1}^T \cp(Y_i|_{Y_{<i} = y_{<i}})$
instead of $(1/T)\sum_{i=1}^T  \cp(h(X_i)|_{X_{<i} = x_{<i}})$.
That is, Lemma~\ref{lem:cp_upper_bound} requires us to condition on previous
{\em hashed values} $Y_{<i}$, whereas the above argument refers to 
conditioning
on the un-hashed values $X_{<i}$.  The difficulty with directly reasoning 
about the former
is that conditioned on the hashed values $Y_{<i}$, the hash function $H$ may 
no longer be uniform
(as it is correlated with $Y_{<i}$) and thus the Leftover Hash Lemma no 
longer applies.

To get around with the issues, we work with the
averaged form of conditional collision
probability $\cp(Y_i|H,Y_{<i})$, as from Definition~\ref{def:cond_cp}.
Our key observation is that now we can apply Lemma~\ref{lem:cond_cp}
to deduce that for every block $i \in [T]$, the
conditional collision probability satisfies
$\cp(Y_i|H,Y_{<i}) \leq \cp(Y_i|H,X_{<i}) \leq
1/M + 1/K$.  Then, by a Markov argument, it follows that with
probability $1-\eps$ over $(h,\vecy)\leftarrow (H,\vecY)$, the
average collision probability satisfies
$$\frac{1}{T}\sum_{i=1}^T \cp(Y_i|_{(H,Y_{<i}) = (h,y_{<i})}) \leq \frac{1}
{M} + \frac{1}{K\eps}.$$
We can then modify an $\eps$-fraction of distribution, and apply
Lemma \ref{lem:cp_upper_bound} to complete the analysis.

The following lemma formalizes our claim about that the conditional collision probability
of every block of $(H,\vecY)$ is small.

\begin{lemma} \label{lem:small_cond_cp}
  Let $H:[N]\rightarrow[M]$ be a random hash function from a
  $2$-universal family $\Hfam$. Let $\vecX = (X_1,\dots,X_T)$ be a
  block $K$-source over $[N]^T$. Let $(H,\vecY) = (H,H(X_1),\dots,H(X_T))$.
  Then $\cp(H) = 1/|\Hfam|$ and for every $i\in [T]$, $\cp(Y_i|H,Y_{<i}) 
\leq 1/M + 1/K$.
\end{lemma}
\begin{proof}
  $\cp(H) = 1/|\Hfam|$ is trivial since $H$ is the uniform distribution.
  Fix $i \in [T]$. By the definition of block $K$-source, for every $x_{<i}$
  in the support of $X_{<i}$, $\cp(X_i|_{X_{<i} = x_{<i}}) \leq
  1/K$. By the Leftover Hash Lemma, we have $\cp((Y_i|_{X_{<i} =
  x_{<i}})|(H|_{X_{<i} = x_{<i}}) ) \leq 1/M + 1/K$ for every $x_{<i}$. It 
follows that
  $\cp(Y_i|H,X_{<i}) \leq 1/M + 1/K$. Now, we can think of
  $(Y_i|H,X_{<i})$ as $Y_i$ first conditioning on $(H,Y_{<i})$, and then
  further conditioning on $X_{<i}$. By Lemma \ref{lem:cond_cp},
  we have $$\cp(Y_i|H,Y_{<i}) \leq \cp(Y_i|H,Y_{<i},X_{<i}) = \cp(Y_i|
H,X_{<i}) \leq
  1/M + 1/K,$$
  as desired.
\end{proof}

\ifnum\short=1
The remaining part of the proof follows the above sketch closely. Details can be found in the full version of this paper\cite{ChungVa08TR}.
\fi

\ifnum\short=0
We use this to prove Theorem \ref{thm:2-univ_cp} as outlined above.

\begin{proofof}{of Theorem \ref{thm:2-univ_cp}}
By Lemma \ref{lem:small_cond_cp}, for every $i \in [T]$, we have
$$\E_{(h,\vecy)\leftarrow (H,\vecY)}\left[\cp(Y_i|_{(H,Y_{<i}) = 
(h,y_{<i})})\right]
= \cp(Y_i|H,Y_{<i}) \leq \frac{1}{M} + \frac{1}{K}.$$ By linearity
of expectation, the average conditional collision probability is
also small.
\begin{equation*}
\E_{(h,\vecy)\leftarrow (H,\vecY)}\left[\frac{1}{T} \sum_{i=1}^T
\cp(Y_i|_{(H,Y_{<i}) = (h,y_{<i})})\right] \leq\frac{1}{M} +
\frac{1}{K}.
\end{equation*}
Note that the collision probability of a random variable over $[M]$
is at least $1/M$. Thus, Markov's inequality implies that with
probability at least $1-\eps$ over $(h,\vecy)\leftarrow (H,\vecY)$,
\begin{equation}\label{eq:Markov}
\frac{1}{T} \sum_{i=1}^T \cp(Y_i|_{(H,Y_{<i}) =(h,y_{<i})}) \leq
\frac{1}{M} + \frac{1}{K\eps} = \frac{1}{M}\cdot
\left(1+\frac{M}{K\eps}\right).
\end{equation}
In Lemma \ref{lem:modify} below, we show how to fix the \emph{bad}
$(h,\vecy)$'s by modifying at most $\eps$-fraction of the
distribution. Formally, Lemma \ref{lem:modify} says that there
exists a distribution $(H,\vecZ) = (H,Z_1,\dots,Z_T)$ such that
$(H,\vecY)$ is $\eps$-close to $(H,\vecZ)$, and for every
$(h,\vecz)\leftarrow (H,\vecZ)$,
$$\frac{1}{T} \sum_{i=1}^T \cp(Z_i|_{(H,Z_{<i}) =(h,z_{<i})}) \leq \frac{1}
{M}\cdot
\left(1+\frac{M}{K\eps}\right).$$ Applying  Lemma
\ref{lem:cp_upper_bound} on $(\vecZ|_{H=h})$ for every $h\in
\supp(\Hfam)$, we have
$$\cp(H,\vecZ) = \frac{1}{|\Hfam|} \cdot \E_{h\leftarrow H}\left[\cp(\vecZ|
_{H=h})\right]  \leq \frac{1}{|\Hfam|\cdot M^T} \cdot
\left(1+\frac{M}{K\eps}\right)^T.$$
\end{proofof}

\begin{lemma} \label{lem:modify}
Let $(H,\vecY) = (H,Y_1,\dots,Y_T)$ be jointly distributed random variables over $\Hfam \times [M]^T$ such that with probability at least $1- \eps$ over $(h,\vecy) \leftarrow (H,\vecY)$, the average conditional collision probability satisfies
  $$\frac{1}{T}\cdot \sum_{i=1}^T \cp(Y_i|_{(H,Y_{<i}) = (h,y_{<i})}) \leq \frac{1}{M} + \alpha.$$
Then there exists a distribution $(H,\vecZ) = (H,Z_1,\dots,Z_T)$ such that $(H,\vecZ)$ is $\eps$-close to $(H,\vecY)$, and for every $(h,\vecz) \in \supp(H,\vecZ)$, we have
  $$\frac{1}{T}\cdot \sum_{i=1}^T \cp(Z_i|_{(H,Z_{<i}) = (h,z_{<i})}) \leq \frac{1}{M} + \alpha.$$
Furthermore, the marginal distribution of $H$ is unchanged.
\end{lemma}
\begin{proof}
We define the distribution $(H,\vecZ)$ as follows.
\begin{itemize}
  \item Sample $(h,\vecy) \leftarrow (H,\vecY)$.
  \item If $(1/T)\cdot\sum_{i=1}^T \cp(Y_i|_{(H,Y_{<i}) = (h,y_{<i})}) \leq 1/M + \alpha$, then output $(h,\vecy)$.
  \item Otherwise, let $j\in [T]$ be the least index such that
    $$\frac{1}{T}\sum_{i=1}^j \left(\cp(Y_i|_{(H,Y_{<i}) = (h,y_{<i})})- \frac{1}{M}\right) \leq \alpha \mbox{ and  } \frac{1}{T}\sum_{i=1}^{j+1} \left(\cp(Y_i|_{(H,Y_{<i}) = (h,y_{<i})})-\frac{1}{M}\right) >
\alpha$$
  \item Choose $w_{j+1},\dots,w_T \leftarrow U_{[M]}$, and output $(h,y_1,\dots,y_j,w_{j+1},\dots,w_T)$.
\end{itemize}
It is easy to check that (i) $(H,\vecZ)$ is well-defined, (ii) $(H,\vecY)$ is $\eps$-close to $(H,\vecZ)$, (iii) for every $(h,\vecz) \in (H,\vecZ)$, $(1/T)\cdot \sum_{i=1}^T \cp(Z_i|_{(H,Z_{<i}) = (h,z_{<i})}) \leq 1/M + \alpha$, and (iv) the marginal distribution of $H$ is unchanged.
\end{proof}
\fi

\subsection{Small Collision Probability Using $4$-wise Independent
Hash Functions} \label{subsec:4-wise_cp}

As discussed in \cite{MitzenmacherV08}, using $4$-wise independent
hash functions $H:[N] \rightarrow [M]$ from $\Hfam$, we can further
reduce the required randomness in the data $\vecX =
(X_1,\dots,X_T)$. \cite{MitzenmacherV08} shows that in this case, $K
\geq MT + \sqrt{2MT^3/\eps}$ is enough for the hashed sequence
$(H,\vecY)$ to be $\eps$-close to having collision probability
$O(1/|\Hfam|\cdot M^T)$. As discussed in the previous subsection, by
avoiding using union bounds, we show that $K \geq MT +
\sqrt{2MT^2/\eps}$ suffices. (Taking logs yields the second entry in
Table \ref{tbl:analysis}, i.e. it suffices to have Renyi entropy
$k=\max\{m+\log T,(1/2)\cdot(m+2\log T+\log(1/\eps))\}+O(1)$ per
block.) Formally, we prove the following theorem.

\begin{theorem} \label{thm:4-wise_cp}
  Let $H:[N]\rightarrow[M]$ be a random hash function from a
  $4$-wise independent family $\Hfam$. Let $\vecX = (X_1,\dots,X_T)$ be a
  block $K$-source over $[N]^T$. For every $\eps > 0$, the hashed
  sequence
  $(H,\vecY) = (H,H(X_1),\dots,H(X_T))$ is $\eps$-close to a
  distribution $(H,\vecZ) = (H,Z_1,\dots,Z_T)$ such that
  $$\cp(H,\vecZ) \leq \frac{1}{|\Hfam|\cdot M^T}\left(
  1+\frac{M}{K} + \sqrt{\frac{2M}{K^2\eps}}\right)^T.$$
  In particular, if $K \geq MT + \sqrt{2MT^2/\eps}$, then $(H,\vecZ)$ has 
collision
  probability at most $(1+\gamma)/(|\Hfam|\cdot M^T)$ for $\gamma = 2\cdot 
(MT + \sqrt{2MT^2/\eps})/K$.
\end{theorem}


The improvement of Theorem \ref{thm:4-wise_cp} over Theorem
\ref{thm:2-univ_cp} comes from that when we use $4$-wise independent
hash families, we have a concentration result on the conditional
collision probability for each block
\ifnum\short=1
. For the proof of the theorem, please refer to \cite{ChungVa08TR}.
\fi
\ifnum\short=0 
, via the following lemma.

\begin{lemma}[\cite{MitzenmacherV08}] \label{lem:4-wise_variance}
  Let $H:[N]\rightarrow[M]$ be a random hash function from a
  $4$-wise independent family $\Hfam$, and $X$ a random variable
  over $[N]$ with $\cp(X) \leq 1/K$. Then we have
  $$\Var_{h\leftarrow H}[\cp(h(X))] \leq \frac{2}{MK^2}.$$
\end{lemma}

We can then replace the application of Markov's Inequality in the
proof of Theorem~\ref{thm:2-univ_cp} by Chebychev's Inequality to
get stronger result. Formally, we prove the following lemma, which
suffices to prove Theorem \ref{thm:4-wise_cp}.

\begin{lemma} \label{concentration}
  Let $H:[N]\rightarrow[M]$ be a random hash function from a
  $4$-wise independent family $\Hfam$. Let $\vecX = (X_1,\dots,X_T)$ be a
  block $K$-source over $[N]^T$. Let $(H,\vecY) =
  (H,H(X_1),\dots,H(X_T))$. Then with probability at least $1-\eps$
  over $(h,\vecy) \leftarrow (H,\vecY)$,
  \begin{equation*}
    \frac{1}{T} \sum_{i=1}^T \cp(Y_i|_{(H,Y_{<i}) =(h,y_{<i})}) \leq
    \frac{1}{M}\cdot \left(1+\frac{M}{K} + \sqrt{\frac{2M}{K^2\eps}}\right).
  \end{equation*}
\end{lemma}

Theorem \ref{thm:4-wise_cp} follows immediately by composing Lemma
\ref{concentration}, \ref{lem:modify}, and \ref{lem:cp_upper_bound}
in the same way as the proof of Theorem \ref{thm:2-univ_cp}.

\begin{proofof}{of Lemma~\ref{concentration}}
Recall that we have
$$\E_{(h,\vecy)\leftarrow (H,\vecY)}\left[\frac{1}{T} \sum_{i=1}^T
\cp(Y_i|_{(H,Y_{<i})=(h,y_{<i})})\right] \leq \frac{1}{M} +
\frac{1}{K}.$$ Hence, our goal is to upper bound the probability of
the value $(1/T)\sum_{i=1}^T \cp(Y_i|_{(H,Y_{<i})=(h,y_{<i})})$
deviating from its mean by $\sqrt{2/MK^2\eps}$. Our strategy is to
bound the variance of a properly defined random variable, and then
apply Chebychev's Inequality. By Lemma \ref{lem:4-wise_variance},
the information we get from $4$-wise independent hash function is
that for every $i\in[T]$, we have
\begin{equation} \label{eq:small_variance}
\Var_{h\leftarrow H}\left[ \cp(Y_i|_{(H,X_{<i})=(h,x_{<i})})\right]
\leq \frac{2}{MK^2} \quad \quad \forall x_{<i} \in \supp(X_{<i})
\end{equation}

Fix $i \in [T]$, let us try to bound the variance of the $i$-th
block. There are two issues to take care of. Firstly, the variance we
have is conditioning on $X_{<i}$ instead of $Y_{<i}$. Secondly, even
when conditioning on $X_{<i}$, it is possible that the variance is
$$\Var_{(h,\vecx)\leftarrow(H,\vecX)}\left[\cp(Y_i|_{(H,X_{<i})=(h,x_{<i})})\right] =
\Omega\left(\frac{1}{K^2}\right) \gg \frac{2}{MK^2}.$$ 
The reason is that conditioning on different $X_{<i} = x_{<i}$, the collision probability of
$(Y_i|_{X_{<i} = x_{<i}})$ may have different expectation over
$h\leftarrow \Hfam$. Thus, we have to subtract the mean first. Let
us define $$f(h,x_{<i}) = \cp(Y_i|_{(H,X_{<i})=(h,x_{<i})}) -
\E_{h\leftarrow H}\left[\cp(Y_i|_{(H,X_{<i})=(h,x_{<i})})\right]$$
Now, for every $x_{<i}\in \supp(X_{<i})$, $f(H,x_{<i})$ has mean
$0$, and variance $\leq 2/MK^2$. It follows that
$$\Var_{(h,\vecx)\leftarrow(H,\vecX)} \left[ f(h,x_{<i})\right] \leq 
\frac{2}{MK^2}.$$
We now deal with the issue of conditioning on $X_{<i}$ versus
$Y_{<i}$. Let us define
$$g(h,y_{<i}) = \E_{x_{<i} \leftarrow (X_{<i}|_{(H,Y_{<i})=(h,y_{<i})})}
\left[
f(h,x_{<i})\right].$$ We claim that
$$\cp(Y_i|_{(H,Y_{<i})=(h,y_{<i})}) \leq \frac{1}{M} + \frac{1}{K} + 
g(h,y_{<i}).$$
Indeed, by Lemma \ref{lem:cond_cp} and the definition of $f$ and $g$,
\begin{eqnarray*}
\lefteqn{\cp(Y_i|_{(H,Y_{<i})=(h,y_{<i})})}\\ &\leq&
\cp((Y_i|_{(H,Y_{<i})=(h,y_{<i})})| (X_i|_{(H,Y_{<i})=(h,y_{<i})}))
\\
& = & \E_{x_{<i} \leftarrow (X_{<i}|_{(H,Y_{<i})=(h,y_{<i})})}\left[
\cp(Y_i|_{(H,X_{<i})=(h,x_{<i})})\right] \\
& = & \E_{x_{<i} \leftarrow (X_{<i}|_{(H,Y_{<i})=(h,y_{<i})})}\left[
f(h,x_{<i})+ \E_{h\leftarrow
H}\left[\cp(Y_i|_{(H,X_{<i})=(h,x_{<i})})\right] \right] \\
& \leq &  g(h,y_{<i})+ \frac{1}{M} + \frac{1}{K}
\end{eqnarray*}
Also note that $g(H,Y_{<i})$ has mean $0$ and small variance:
$$\E_{(h,y_{<i})\leftarrow(H,y_{<i})} [g(h,y_{<i})] = \E_{(h,\vecx)
\leftarrow(H,\vecX)} [ f(h,x_{<i})] =0, $$
$$\Var_{(h,y_{<i})\leftarrow(H,y_{<i})} [g(h,y_{<i})] \leq  \Var_{(h,\vecx)
\leftarrow(H,\vecX)} [ f(h,x_{<i})] \leq \frac{2}{MK^2}.$$
The above argument holds for every block $i\in [T]$. Taking average
over blocks, we get
$$\E_{(h,\vecy)\leftarrow(H,\vecY)} \left[ \frac{1}{T} \sum_{i=1}^T 
g(h,y_{<i})\right]  =0, $$
$$\Var_{(h,\vecy)\leftarrow(H,\vecY)} \left[ \frac{1}{T} \sum_{i=1}^T 
g(h,y_{<i})\right] \leq   \frac{2}{MK^2}, \mbox{ and }$$
$$\frac{1}{T} \sum_{i=1}^T\cp(Y_i|_{(H,Y_{<i})=(h,y_{<i})}) \leq \frac{1}{M} 
+ \frac{1}{K} + \left(\frac{1}{T} \sum_{i=1}^Tg(h,y_{<i})\right).$$
Finally, we can apply Chebychev's Inequality to random variable
$(1/T)\cdot \sum_i g(H,Y_{<i})$ to get the desired result: with
probability $1-\eps$ over $(h,\vecy) \leftarrow (H,\vecY)$,
  \begin{equation*}
    \frac{1}{T} \sum_{i=1}^T \cp(Y_i|_{(H,Y_{<i}) =(h,y_{<i})}) \leq
    \frac{1}{M}\cdot \left(1+\frac{M}{K} + \sqrt{\frac{2M}{K^2\eps}}\right).
  \end{equation*}
\end{proofof}

\fi

\subsection{Statistical Distance to Uniform Distribution}

Let $H:[N]\rightarrow [M]$ be a random hash function form a $2$-universal family $\Hfam$. Let $\vecX = (X_1,\dots,X_T)$ be a block $K$-source over $[N]^T$. In this subsection, we study the statistical distance between the distribution of hashed sequence $(H,\vecY) = (H,H(X_1),\dots,H(X_T))$ and the uniform distribution
$(H,U_{[M]^T})$.  Classic results of \cite{ChorG88,ImpagliazzoLL89,Zuckerman96} show that if $K \geq
MT^2/\eps^2$, then $(H,\vecY)$ is $\eps$-close to uniform. The proof idea is as follows. The Leftover Hash Lemma together with Lemma \ref{lem:cp_to_stat} tells us that the joint distribution of hash function and a hashed value $(H,Y_i) = (H,H(X_i))$ is $\sqrt{M/K}$-close to uniform $U_{[M]}$ even conditioning on the previous blocks $X_{<i}$. One can then use a hybrid argument to show that the distance grows linearly with the number of blocks, so $(H,\vecY)$ is $T\cdot \sqrt{M/K}$-close to uniform. Taking $K \geq MT^2/\eps^2$ completes the analysis.

We save a factor of $T$, and show that in fact, $K = MT/\eps^2$ is
sufficient.  (Taking logs yields the third entry in Table 
\ref{tbl:analysis}, i.e. it
suffices to have Renyi entropy $k=m+\log T+2\log(1/\eps)$ per block.)
Formally, we prove the following theorem.

\begin{theorem} \label{thm:2-univ_stat}
  Let $H:[N]\rightarrow[M]$ be a random hash function from a
  $2$-universal family $\Hfam$. Let $\vecX = (X_1,\dots,X_T)$ be a
  block $K$-source over $[N]^T$. For every $\eps > 0$ such that $K > MT/
\eps^2$,
  the hashed sequence $(H,\vecY) = (H,H(X_1),\dots,H(X_T))$ is $\eps$-close 
to
  uniform $(H,U_{[M]^T})$.
\end{theorem}

Recall that the previous analysis goes by passing to statistical distance first, and then measuring the growth of distance using statistical distance. This incurs a quadratic dependency of $K$ on $T$. Since without further information, the hybrid argument is tight, to save a factor of $T$, we have to measure the increase of distance over blocks in another way, and pass to statistical distance only in the end. It turns out that the {\em Hellinger distance} (cf., \cite{GibbsSu02}) is a good measure for our purposes: 

\begin{definition}[Hellinger distance] \label{def:Hellinger}
  Let $X$ and $Y$ be two random variables over $[M]$.
  The \emph{Hellinger distance} between $X$ and $Y$ is
  $$d(X,Y) \eqdef \left(\frac{1}{2} \sum_i
  (\sqrt{\Pr[X=i]}-\sqrt{\Pr[Y=i]})\right)^{1/2} = \sqrt{1-\sum_i 
\sqrt{\Pr[X=i]\cdot \Pr[Y=i]}}.$$
\end{definition}

Like statistical distance, Hellinger distance is a distance measure
for distributions, and it takes value in $[0,1]$. The following
standard lemma 
says that the two distance measures are closely related. We
remark that the lemma is tight in both directions even if $Y$ is the
uniform distribution.

\begin{lemma}[cf., \cite{GibbsSu02}] \label{lem:Hellinger_to_stat}
  Let $X$ and $Y$ be two random variables over $[M]$.
  We have $$d(X,Y)^2 \leq \Delta(X,Y) \leq \sqrt{2} \cdot d(X,Y).$$
\end{lemma}

In particular, the lemma allows us to upper-bound the statistical
distance by upper-bounding the Hellinger distance. Since our goal is
to bound the distance to uniform, it is convenient to introduce the
following definition.

\begin{definition}[Hellinger Closeness to Uniform]
  Let $X$ be a random variable over $[M]$. The \emph{Hellinger closeness} 
of  $X$ to uniform $U_{[M]}$ is  $$C(X) \eqdef \frac{1}{M} \sum_i \sqrt{M\cdot \Pr[X=i]} = 1-
d(X,U_{[M]})^2.$$
\end{definition}

Note that $C(X,Y) = C(X) \cdot
C(Y)$ when $X$ and $Y$ are independent random variables, so
the Hellinger closeness is well-behaved with respect to products (unlike
statistical distance).  By Lemma
\ref{lem:Hellinger_to_stat}, if the Hellinger closeness $C(X)$ is close to
$1$, then $X$ is close to uniform in statistical distance. Recall
that collision probability behaves similarly. If the collision
probability $\cp(X)$ is close to $1/M$, then $X$ is close to
uniform. In fact, by the following normalization, we can view the
collision probability as the $2$-norm of $X$, and the Hellinger closeness as
the $1/2$-norm of $X$.

Let $f(i) = M \cdot \Pr[X = i]$ for $i \in [M]$. In terms of
$f(\cdot)$, the collision probability is $\cp(X) = (1/M^2)\cdot
\sum_i f(i)^2$, and Lemma \ref{lem:cp_to_stat} says that if the
``$2$-norm" $M\cdot \cp(X) = \E_{i}[f(i)^2] \leq 1+\eps$  where the
expectation is over uniform $i \in [M]$, then $\Delta(X,U) \leq
\sqrt{\eps}$,. Similarly, Lemma \ref{lem:Hellinger_to_stat} says
that if the ``$1/2$-norm" $C(X) = \E_i[\sqrt{f(i)}] \geq 1-\eps$,
then $\Delta(X,U) \leq \sqrt{\eps}$.

We now discuss our approach to prove Theorem \ref{thm:2-univ_stat}.
We want to show that $(H,\vecY)$ is close to uniform. All we know is
that the conditional collision probability $\cp(Y_i|H,Y_{<i})$ is
close to $1/M$ for every block.  If all blocks are
independent, then the overall collision probability $\cp(H,\vecY)$
is small, and so $(H,\vecY)$ is close to uniform. However, this is
not true without independence, since $2$-norm tends to over-weight
heavy elements. In contrast, the $1/2$-norm does not suffer this
problem. Therefore, our approach is to show that small conditional
collision probability implies large Hellinger closeness. Formally, we have
the following lemma.
\ifnum\short=1
The main idea is to use H\"older's inequality to relate two different norms. 
\fi

\begin{lemma} \label{lem:cp_to_closeness}
  Let $\vecX = (X_1,\dots,X_T)$ be jointly distributed random
  variables over $[M_1]  \times \dots \times [M_T]$ such that $\cp(X_i|
X_{<i}) \leq
  \alpha_i/ M_i$ for every $i\in[T]$. Then the Hellinger closeness satisfies
  $$C(\vecX) \geq
  \sqrt{\frac{1}{\alpha_1  \dots  \alpha_T}}.$$
\end{lemma}

\ifnum\short=1 
The proof of this lemma can be found in the full version of this paper\cite{ChungVa08TR}. 
\fi 
With this lemma, the proof of Theorem \ref{thm:2-univ_stat} is
immediate.  


\begin{proofof}{of Theorem \ref{thm:2-univ_stat}}
By Lemma \ref{lem:small_cond_cp}, $\cp(H) = 1/|\Hfam|$, and
$\cp(Y_i|H,Y_{<i}) \leq (1+M/K)/M$ for every $i\in [T]$. By Lemma
\ref{lem:cp_to_closeness}, the Hellinger closeness satisfies $C(H,\vecY) \geq
(1+M/K)^{-T/2} \geq 1- MT/2K$ (recall that $K \geq MT/\eps^2$). It
follows by Lemma \ref{lem:Hellinger_to_stat} that
$$\Delta((H,\vecY), (H,U_{[M]^T})) \leq  \sqrt{2} \cdot d((H,\vecY),
(H,U_{[M]^T})) = \sqrt{2}\cdot \sqrt{1-C(H,\vecY)} \leq \sqrt{MT/K}
\leq \eps.$$
\end{proofof}

\ifnum\short=0
We proceed to prove Lemma \ref{lem:cp_to_closeness}. The main idea is to use H\"older's inequality to relate two different norms. We recall H\"older's inequality first.

\begin{lemma}[H\"older's inequality\cite{Durrett04}] \ \label{lem:Holder}
\begin{itemize}
\item Let $F, G$ be two non-negative functions from $[M]$ to $\R$,
and $p,q > 0$ satisfying $1/p + 1/q = 1$. Let $x$ be a uniformly
random index over $[M]$. We have
$$\E_x[F(x)\cdot G(x)] \leq \E_x[F(x)^p]^{1/p} \cdot \E_x[G(x)^q]^{1/q}.$$

\item In general, let $F_1, \dots, F_n$ be non-negative functions
from $[M]$ to $\R$, and $p_1,\dots p_n > 0$ satisfying $1/p_1 +
\dots 1/p_n = 1$. We have
$$\E_x[F_1(x)\cdots F_n(x)] \leq \E_x[F_1(x)^{p_1}]^{1/p_1} \cdots
\E_x[F_n(x)^{p_n}]^{1/p_n}.$$
\end{itemize}
\end{lemma}

\begin{proofof}{of Lemma \ref{lem:cp_to_closeness}}
We prove it by induction on $T$. The base case $T=1$ is already
non-trivial. Let $X$ be a random variable over $[M]$ with $\cp(X)
\leq \alpha/M$, we need to show that the Hellinger closeness $C(X) \geq
\sqrt{1/\alpha}$. Recall the normalization we mentioned before. Let
$f(x) = M \cdot \Pr[X=x]$ for every $x\in[M]$. In terms of
$f(\cdot)$, we want to show that $\E_x[f(x)^2] \leq \alpha$ implies
$\E_x[\sqrt{f(x)}] \geq \sqrt{1/\alpha}$. Note that $\E_x[f(x)] =
1$. We now apply H\"older's inequality with $F = f^{2/3}$, $G =
f^{1/3}$, $p = 3$, and $q = 3/2$. We have
$$ \E_x[f(x)] \leq \E_x[f(x)^2]^{1/3} \cdot \E_x[f(x)^{1/2}]^{2/3},$$
which implies
$$C(X) = \E_x[\sqrt{f(x)}] \geq \E_x[f(x)]^{3/2} / \E_x[f(x)^2]^{1/2} \geq 
\sqrt{1/\alpha}.$$

Suppose the lemma is true for $T-1$, we show that it is true for
$T$. Let $f(x) = M_1\cdot \Pr[X_1 = x]$. To apply the induction
hypothesis, we consider the conditional random variables
$(X_2,\dots,X_T|_{X_1 = x})$ for every $x \in [M_1]$. For every
$x\in[M_1]$ and $j = 2,\dots, T$, we define $g_j(x) = M_j \cdot
\cp((X_j|_{X_1 = x})|(X_2,\dots,X_{j-1}|_{X_1 = x}))$ to be the
"normalized" conditional collision probability. By induction
hypothesis, we have $C(X_2,\dots,X_T|_{X_1 = x}) \geq
\sqrt{g_2(x)\cdots g_T(x)}$ for every $x \in [M_1]$. It follows that
$$C(\vecX) = \E_x[\sqrt{f(x)}\cdot C(X_2,\dots,X_T|_{X_1 = x)}] \geq 
\E_x[\sqrt{f(x)/g_2(x)\cdots g_T(x)}].$$

We use H\"older's inequality twice to show that
$\E_x[\sqrt{f(x)/g_2(x)\cdots g_T(x)}] \geq
\sqrt{1/\alpha_1\cdots\alpha_T}$. Let us first summarize the
constraints we have. By definition, we have $\E_x[f(x)^2] \leq
\alpha_1$. Fix $j \in \{ 2,\dots,T\}$. Note that
\begin{eqnarray*}
\lefteqn{\cp(X_j|X_{<j})} \\
& = & \E_{x\leftarrow X_1}\left[\cp((X_j|_{X_1 = x})| (X_2,\dots,X_{j-1}|_{X_1 = x}))\right] \\
& = & \E_{x\leftarrow X_1}[g_j(x)/M_j] \\
& = & \E_{x\leftarrow U_{[M]}}[f(x)g_j(x)]/M_j
\end{eqnarray*}
It follows that $\E_x[f(x)g_j(x)] \leq
\alpha_j$ for $j = 2,\dots, T$. Now, we apply the second version of
H\"older's Inequality with $F_1 = (f/g_2\cdots g_T)^{1/2}$, $F_j =
(fg_j)^{1/(T+1)}$ for $j=2,\dots,T$, $p_1 = 2/(T+1)$, and $p_j =
1/(T+1)$ for $j=2,\dots,T$, which gives
$$\E_x\left[f(x)^{T/(T+1)}\right] \leq \E_x\left[\sqrt{f(x)/g_2(x)\cdots
g_T(x)}\right]^{2/(T+1)}\cdot \E_x\left[f(x)g_2(x)\right]^{1/(T+1)}
\cdots \E_x\left[f(x)g_T(x)\right]^{1/(T+1)},$$ so
\begin{eqnarray*}
\E_x\left[\sqrt{f(x)/g_2(x)\cdots g_T(x)}\right] & \geq
&\E_x\left[f(x)^{T/(T+1)}\right]^{(T+1)/2} \cdot \prod_{j=2}^T
\E_x\left[f(x)g_j(x)\right]^{-1/2} \\
 &\geq & \E_x\left[f(x)^{T/(T+1)}\right]^{(T+1)/2} \cdot 
\sqrt{1/\alpha_2\cdots\alpha_T}.
\end{eqnarray*}
It remains to lower bound the first term by $\sqrt{1/\alpha_1}$. We
apply H\"older again with $F = f^{2/(T+2)}$, $G = f^{T/(T+2)}$,
$p=T+2$, and $q = (T+2)/(T+1)$, which gives
$$\E_x\left[f(x)\right] \leq \E_x\left[f(x)^2\right]^{1/(T+2)} \cdot
\E_x\left[f(x)^{T/(T+1)}\right]^{(T+1)/(T+2)},$$ so
$$ \E_x\left[f(x)^{T/(T+1)}\right]^{(T+1)/2} \geq
\E_x\left[f(x)\right]^{(T+2)/2}/\E_x\left[f(x)^2\right]^{1/2} \geq
\sqrt{1/\alpha_1}.$$ Combining the inequalities, we have $C(\vecX)
\geq \sqrt{1/\alpha_1\cdots\alpha_T}$.
\end{proofof}
\fi

\section{Negative Results: How Much Entropy is Necessary?}
\label{sec:negative}

In this section, we provide lower bounds on the entropy needed for
the data items.  We show that if $K$ is not large enough, then for
every hash family $\Hfam$, there exists a block $K$-source $\vecX =
(X_1,\dots,X_T)$ such that the hashed sequence $\vecY = (
H(X_1),\dots,H(X_T))$ do not satisfy the desired closeness
requirements to uniform (possibly in conjunction with the hash
function $H$). 
\ifnum\short=1
Due to space constraints, we only discuss the proof idea of lower bounds in this section. Please refer to the full version of this paper\cite{ChungVa08TR} for the complete proofs.
\fi

\subsection{Lower Bound for Statistical Distance to Uniform Distribution}

Let us first consider the requirement for the joint distribution of
$(H,\vecY)$ being $\eps$-close to uniform. When there is only one
block, this is exactly the requirement for a ``strong extractor''. The
lower bound in the extractor literature, due to Radhakrishnan and Ta-Shma~\cite{RadhakrishnanT00}
shows that $K \geq \Omega(M/\eps^2)$ is necessary, which is tight up
to a constant factor. Our goal is to show that when hashing $T$
blocks, the value of $K$ required  for each block increases by a factor of $T$. 
Intuitively, each block will produce some error (i.e., the hashed
value is not close to uniform), and the overall error will
accumulate over the blocks, so we need to inject more randomness per
block to reduce the error. Indeed, we use this intuition to show
that  $K \geq \Omega(MT/\eps^2)$ is necessary for the hashed
sequence to be $\eps$-close to uniform, matching the upper bound in
Theorem \ref{thm:2-univ_stat}. Note that the lower bound holds even
for a truly random hash family. Formally, we prove the following
theorem.

\begin{theorem} \label{thm:lb_stat}
Let $N,M,$ and $T$ be positive integers and $\eps \in (0,\eps_0)$ a real
number such that $N \geq MT/\eps^2$, where $\eps_0 > 0$ is a small
absolute constant. Let $H:[N] \rightarrow [M]$ be a random hash
function from an hash family $\Hfam$. Then there
  exists an integer $K = \Omega(MT/\eps^2)$, and a 
block $K$-source
  $\vecX =(X_1,\dots,X_T)$ such that $(H,\vecY) = (H,H(X_1),\dots,H(X_T))$
  is $\eps$-far from uniform $(H,U_{[M]^T})$ in statistical distance.
\end{theorem}

To prove the theorem, we need to find such an $\vecX$ for every hash family $\Hfam$. Following the intuition, we find an $X$ that incurs certain error on a single block, and take $\vecX = (X_1,\dots,X_T)$ to be $T$ i.i.d. copies of $X$. More precisely, we first find a $K$-source $X$ such that for $\Omega(1)$-fraction of hash functions $h \in \Hfam$, $h(X)$ is $\Omega(\eps/\sqrt{T})$-far from uniform. This step is the same as the lower bound proof for extractors~\cite{RadhakrishnanT00}, which uses the probabilistic method. We pick $X$ to be a random \emph{flat} $K$-source, i.e., a
uniform distribution over a random set of size $K$, and show that $X$ satisfies the desired property with nonzero probability. The next step is to measure how the error accumulates over independent blocks. Note that for a fixed hash function $h$, the hashed sequence $(h(X_1),\dots,h(X_T))$ consists of $T$ i.i.d. copies of $h(X)$. Reyzin \cite{Reyzin04} has shown that the statistical distance increases $\sqrt{T}$ when we have $T$ independent copies for small $T$. However, Reyzin's result only shows an increase up to distance $O(\delta^{1/3})$, where $\delta$ is the
statistical distance of the original random variables. We improve Reyzin's result to show that the $\Omega(\sqrt{T})$ growth continues until the distance reaches some absolute constant. We then use it to show that the joint distribution $(H,\vecY)$ is far from uniform. 
\ifnum\short=1
Details can be found in the full version of this paper\cite{ChungVa08TR}.
\fi

\ifnum\short=0 

The following lemma corresponds to the first step.

\begin{lemma} \label{lem:lb_single_hash}
Let $N$ and $M$ be positive integers and $\eps \in (0,1/4), \delta \in (0,1)$  real numbers such that $N \geq M/\eps^2$. Let $H:[N] \rightarrow [M]$ be a random hash function from an hash family $\Hfam$. Then there exists an integer $K = \Omega(\delta^2 M/\eps^2)$, and a flat $K$-source $X$ over $[N]$ such that with probability at least $1-\delta$ over $h \leftarrow H$, $h(X)$ is $\eps$-far from uniform.
\end{lemma}

\begin{proof}
Let $K = \lfloor \min \{ \alpha \cdot M / \eps^2, N/2\} \rfloor$ for some $\alpha$ to be determined later. Let $X$ be a random flat $K$-source over $[N]$. That is, $X = U_S$ where $S \subset [N]$ is a uniformly random size $K$ subset of $[N]$. We claim that for every hash function $h:[N] \rightarrow [M]$,
  \begin{equation} \label{eq:far_from_unif}
  \Pr_S[ \mbox{ $h(U_S)$ is $\eps$-far from uniform }] \geq 1- c\cdot \sqrt{\alpha}
  \end{equation}
for some absolute constant $c$. Let us assume (\ref{eq:far_from_unif}), and prove the lemma first. Since the claim holds for every hash function $h$,
  $$\Pr_{h\leftarrow H,S}[ \mbox{ $h(U_S)$ is $\eps$-far from uniform }] \geq 1- c\cdot \sqrt{\alpha}.$$ 
Thus, there exists a flat $K$-source $U_S$ such that
  $$\Pr_{h\leftarrow H}[ \mbox{ $h(U_S)$ is $\eps$-far from uniform }] \geq 1- c\cdot \sqrt{\alpha}.$$ 
The lemma follows by setting $\alpha = \min \{ \delta^2/c^2, 1/32 \}$. We proceed to prove (\ref{eq:far_from_unif}). It suffices to show that for every $y\in [M]$, with probability at least $1- c'\cdot \sqrt{\alpha}$ over random $U_S$, the deviation of $\Pr[h(U_S) = y]$ from $1/M$ is at least $4\eps/M$, where $c'$ is another absolute constant. That is,
  \begin{equation}   \label{eq:large_dev}
  \Pr_S\left[ \left|\Pr[h(U_S)=y] - \frac{1}{M}\right| \geq \frac{4\eps}{M} \right] \geq 1 - c' \cdot \sqrt{\alpha}.
  \end{equation}
Again, let us see why (\ref{eq:large_dev}) is sufficient to prove (\ref{eq:far_from_unif}) first. Let us call $y\in[M]$ is \emph{bad} for $S$ if
  $$\left|\Pr[h(U_S)=y] - \frac{1}{M}\right| \geq \frac{4\eps}{M}.$$ 
Since Inequality (\ref{eq:large_dev}) holds for every $y\in[M]$, we have
  \begin{equation*}
  \Pr_{S,y}[ \mbox{$y$ is bad for $S$} ] \geq 1 - c' \cdot \sqrt{\alpha},
  \end{equation*}
where $y$ is uniformly random over $[M]$. It follows that
  $$\Pr_{S}[ \mbox{at least $1/2$-fraction of $y$ are bad for $S$} ] \geq 1 - 2c' \cdot \sqrt{\alpha}$$ 
Observe that if at least $1/2$-fraction of $y$ are bad for $S$, then $\Delta(h(X),U_{[M]}) \geq \eps$. Inequality (\ref{eq:far_from_unif}) follows by setting $c = 2c'$.

It remains to prove (\ref{eq:large_dev}). Let $T = h^{-1}(y)$. We have $\Pr_S[h(U_S)=y] = |S \cap T| /|S|$. Thus, recall that $K \leq \alpha M/\eps^2$, (\ref{eq:large_dev}) follows from  inequality
  $$\Pr_S\left[ \left|  |S \cap T| - \frac{K}{M}\right| < \frac{4K\eps}{M}\right] \leq c' \cdot \sqrt{\frac{K\eps^2}{M}},$$
which follows by the claim below by setting $L = K/M$, and $\beta = 4\eps\sqrt{K/M}$ (Working out the parameters, we have $c' = 4c''$, $\eps < 1/4$ implies $\beta < \sqrt{L}$, and $\alpha \leq 1/32$ implies $\beta < 1$.) 
\begin{claim} \label{clm:binomial}
Let $N, K > 1$ be positive integers such that $N > 2K$, and $L\in [0, K/2]$, $\beta \in (0, \min\{1,\sqrt{L}\})$ real numbers. Let $S \subset [N]$ be a random subset of size $K$, and $T \subset [N]$ be a fixed subset of arbitrary size. We have
$$ \Pr_S \left[ \left| |S\cap T| - L\right| \leq \beta \sqrt{L} \right] \leq c'' \cdot \beta,$$
for some absolute constant $c''$.
\end{claim}
Intuitively, the probability in the claim is maximized when the set $T$ has size $NL/K$ so that $L = \E_S[|S \cap T|]$, and the claim follows by observing that  in this case, the distribution has deviation $\Theta(\sqrt{L})$, and each possible outcome has probability $O(\sqrt{1/L})$. The formal proof of the claim is in Appendix \ref{app:sec:Bino} and is proved by expressing the probability in terms of binomial coefficients, and estimating them using Stirling formula.
\end{proof}

The next step is to measure the increase of statistical distance
over independent random variables. 

\begin{lemma} \label{lem:stat_inc_small_T}
  Let $X$ and $Y$ be random variables over $[M]$ such that
  $\Delta(X,Y) \geq \eps$. Let $\vecX = (X_1,\dots,X_T)$ be $T$
  i.i.d. copies of $X$, and let $\vecY = (Y_1,\dots,Y_T)$ be $T$
  i.i.d. copies of $Y$. We have
  $$ \Delta(\vecX,\vecY) \geq \min\{ \eps_0, c\sqrt{T}\cdot\eps\},$$
  where $\eps_0, c$ are absolute constants.
\end{lemma}

We defer the proof of the above lemma to Appendix \ref{app:proof_lem_stat_inc_small_T}.

\begin{proofof}{of Theorem \ref{thm:lb_stat}}
The absolute constant $\eps_0$ in the theorem is a half of the $\eps_0$ in Lemma \ref{lem:stat_inc_small_T}. By Lemma \ref{lem:lb_single_hash} there is a flat $K$-source such that for $1/2$-fraction of hash functions $h\in \Hfam$, $h(X)$ is $(2\eps/c\sqrt{T})$-far from uniform, for $K = \Omega( (1/2)^2 M/ (2\eps/c\sqrt{T})^2) = \Omega(MT/\eps^2)$. We set $\vecX = (X_1,\dots,X_T)$ to be $T$ independent copies of $X$. Consider a hash function $h$ such that $h(X)$ is $(2\eps/c\sqrt{T})$-far from uniform. By Lemma \ref{lem:stat_inc_small_T}, $(h(X_1),\dots,h(X_T))$ is $2\eps$-far from uniform. Note that this holds for $1/2$-fraction of hash function $h$. It
follows that
  $$\Delta((H,\vecY),(H,U_{[M]})) = \E_{h\leftarrow H}\left[\Delta((h(X_1),\dots,h(X_T), U_{[M]^T})\right]\geq \frac{1}{2} \cdot 2\eps = \eps.$$
\end{proofof}

\fi

\subsection{Lower Bound for Small Collision Probability}

In this subsection, we prove lower bounds on the entropy needed per item to ensure that the sequence of hashed values is close to having small collision probability. Since this requirement is less stringent than being close to uniform, less entropy is needed from the source. The interesting setting in applications is to require the hashed sequence $(H,\vecY) = (H,H(X_1),\dots,H(X_T))$ to be $\eps$-close to having collision probability $O(1/(|\Hfam|\cdot M^T))$. Recall that in this setting, instead of requiring $K \geq MT/\eps^2$, $K \geq \Omega(MT/\eps)$ is sufficient for $2$-universal hash functions (Theorem \ref{thm:2-univ_cp}), and $K \geq \Omega(MT + T\sqrt{M/\eps})$ is sufficient for $4$-wise independent hash functions (Theorem \ref{thm:4-wise_cp}). The main improvement from $2$-universal to $4$-wise independent hashing is the better dependency on $\eps$. Indeed, it can be shown that if we use truly random hash functions, we can reduce the dependency on $\eps$ to $\log (1/\eps)$. Since we are now proving lower bounds for arbitrary hash families, we focus on the dependency on $M$ and $T$. Specifically, our goal is to show that $K = \Omega(MT)$ is necessary. More precisely, we show that when $K \ll MT$, it is possible for the hashed sequence $(H,\vecY)$ to be $.99$-far from any distribution that has collision probability less than $100/(|\Hfam|\cdot M^T)$.

We use the same strategy as in the previous subsection to prove this
lower bound. Fixing a hash family $\Hfam$, we take $T$ independent
copies $(X_1,\dots,X_T)$ of the worst-case $X$ \ifnum\short=0 found in Lemma
\ref{lem:lb_single_hash}\fi, and show that $(H,H(X_1),\dots,H(X_T))$ is
far from having small collision probability. The new ingredient is
to show that when we have $T$ independent copies, and $K\ll MT$,
then $(h(X_1),\dots,h(X_T))$ is very far from uniform (say,
$0.99$-far) for many $h \in \Hfam$. We then argue that in this case,
we can not reduce the collision probability of
$(h(X_1),\dots,h(X_T))$ by changing a small fraction of
distribution, which implies the overall distribution $(H,\vecY)$ is
far from any distribution $(H',\vecZ)$ with small collision
probability. Formally, we prove the following theorem\ifnum\short=0.\fi 
\ifnum\short=1
 \ in the full version of this paper\cite{ChungVa08TR}.
\fi

\begin{theorem} \label{thm:lb_small_cp}
Let $N,M,$ and $T$ be positive integers such that $N \geq MT$. Let $\delta \in (0,1)$ and $\alpha > 1$ be real numbers such that $\alpha < \delta^3 \cdot e^{T/32} / 128$. Let $H:[N] \rightarrow [M]$ be a random hash function from a hash family $\Hfam$. There exists an integer $K = \Omega(\delta^2MT/\log (\alpha/\delta))$, and a block $K$-source $\vecX =(X_1,\dots,X_T)$ such that $(H,\vecY) = (H,H(X_1),\dots,H(X_T))$ is $(1-\delta)$-far from any distribution $(H',\vecZ)$ with $\cp(H',\vecZ) \leq \alpha/(|\Hfam|\cdot M^T)$.
\end{theorem}

Think of $\alpha$ and $\delta$ as constants. Then the theorem says that $K = \Omega(MT)$ is necessary for the hashed sequence $(H,H(X_1),\dots,H(X_T))$ to be close to having small collision probability, matching the upper bound in Theorem \ref{thm:2-univ_cp}.
\ifnum\short=0
In the previous proof, we used Lemma \ref{lem:stat_inc_small_T} to measure the increase of distance over blocks. However, the lemma can only measure the progress up to some small constant. It is known that if the number of copies $T$ is larger then $\Omega(1/\eps^2)$, where $\eps$ is the statistical distance of original copy, then the statistical distance goes to $1$ exponentially fast. Formally, we use the following lemma.

\begin{lemma}[\cite{SahaiVa99}] \label{lem:stat_inc_large_T}
  Let $X$ and $Y$ be random variables over $[M]$ such that
  $\Delta(X,Y) \geq \eps$. Let $\vecX = (X_1,\dots,X_T)$ be $T$
  i.i.d. copies of $X$, and let $\vecY = (Y_1,\dots,Y_T)$ be $T$
  i.i.d. copies of $Y$. We have
  $$ \Delta(\vecX,\vecY) \geq 1 - e^{-T\eps^2/2}.$$
\end{lemma}

We remark that Lemma \ref{lem:stat_inc_small_T} and \ref{lem:stat_inc_large_T} are incomparable. In the parameter range
of Lemma \ref{lem:stat_inc_small_T}, Lemma \ref{lem:stat_inc_large_T} only gives $\Delta( \vecX, \vecY) \geq \Omega(T\eps^2)$ instead of $\Omega(\sqrt{T}\eps)$. To argue that the overall distribution is far from having small collision probability, we introduce the following notion of nonuniformity. 

\begin{definition}
Let $X$ be a random variable over $[M]$ with probability mass function $p$. $X$ is \emph{$(\delta,\beta)$-nonuniform} if for every function $q: [M] \rightarrow \R$ such that $0 \leq q(x) \leq p(x)$ for all $x\in[M]$, and $\sum_x q(x) \geq \delta$, the function satisfies
  $$\sum_{x\in[M]} q(x)^2 > \beta/M.$$
\end{definition}

Intuitively, a distribution $X$ over $[M]$ is $(\delta,\beta)$-nonuniform means that even if we remove $(1-\delta)$-fraction of probability mass from $X$, the ``collision probability'' remains greater than $\beta/M$. In particular, $X$ is $(1-\delta)$-far from any random variable $Y$ with $\cp(Y) \leq \beta/M$.  

\begin{lemma} \label{lem:hard_to_reduce_cp}
Let $X$ be a random variable over $[M]$. If $X$ is $(1-\eta)$-far from uniform, then $X$ is $(2\sqrt{\beta\cdot\eta}, \beta)$-nonuniform for every $\beta \geq 1$. 
\end{lemma}
\begin{proof}
Let $p$ be the probability mass function of $X$, and $q: [M] \rightarrow \R$ be a function such that $0 \leq q(x) \leq p(x)$ for every $x \in [M]$, and $\sum_x q(x) \geq 2\sqrt{\beta\cdot \eta}$. Our goal is to show that $\sum_x q(x)^2 > \beta/M$. Let $T = \{ x\in[M]: p(x) \geq 1/M\}$. Note that $$ \Delta(X,U_{[M]}) = \Pr[X\in T] - \Pr[U_{[M]} \in T] \geq 1-\eta.$$ This implies $\Pr[X\in T] \geq 1-\eta$, and $\mu(T) = \Pr[U_{[M]} \in T] \leq \eta$. Now,
  $$ \sum_{x\in T} q(x) \geq 2\sqrt{\beta \cdot \eta} - \Pr[X \notin T] \geq  2\sqrt{\beta \cdot \eta} - \eta >    \sqrt{\beta \cdot \eta},$$ 
and $\mu(T) \leq \eta$ implies
  $$\sum_{x\in[M]} q(x)^2 \geq \sum_{x\in T} q(x)^2 \geq \frac{ \left(\sum_{x\in T} q(x)\right)^{2}}{|T|}> \frac{\beta}{M}.$$
\end{proof}

We are ready to prove Theorem \ref{thm:lb_small_cp}.

\begin{proofof}{of Theorem \ref{thm:lb_small_cp}}
By Lemma \ref{lem:lb_single_hash} with $\eps = \sqrt{2\ln(128 \alpha/\delta^{3})/T} < 1/4$, there is a flat $K$-source $X$ such that for $(1-\delta/4)$-fraction of hash function $h \in \Hfam$, $h(X)$ is $\eps$-far from uniform, for $K = \Omega( (\delta/4)^2 M /\eps^2) = \Omega( \delta^2 MT/\log(\alpha/\delta))$. We set $\vecX = (X_1,\dots,X_T)$ to be $T$ independent copies of $X$. Consider a hash function $h$ such that $h(X)$ is $\eps$-far from uniform. By Lemma \ref{lem:stat_inc_large_T}, $(h(X_1),\dots,h(X_T))$ is $(1-\eta)$-far from uniform, for $\eta = e^{-\eps^2 T/2} = \delta^3 / 128\alpha$. By Lemma \ref{lem:hard_to_reduce_cp}, $(h(X_1),\dots,h(X_T))$ is $(\delta/4, 2\alpha/\delta)$-nonuniform for $(1-\delta/4)$-fraction of hash functions $h$. By the first statement of Lemma \ref{lem:nonuniform_far} below, this implies  that $(H,\vecY)$ is $(1-\delta)$-far from any distribution $(H',\vecZ)$ with collision probability $\alpha/(|\Hfam|\cdot M^T)$.
\end{proofof}

\begin{lemma} \label{lem:nonuniform_far}
Let $(H,Y)$ be a joint distribution over $\Hfam \times [M]$ such that the marginal distribution $H$ is uniform over $\Hfam$. Let $\eps, \delta, \alpha$ be positive real numbers.
\begin{enumerate}
\item If $Y|_{H=h}$ is $(\delta/4,2\alpha/\delta)$-nonuniform for at least $(1-\delta/4)$-fraction of $h \in \Hfam$, then $(H,Y)$ is $(1-\delta)$-far from any distribution $(H',Z)$ with $\cp(H',Z) \leq \alpha/(|\Hfam|\cdot M)$.
\item If $Y|_{H=h}$ is $(0.1,2\alpha/\eps)$-nonuniform for at least $2\eps$-fraction-frction of $h \in \Hfam$, then $(H,Y)$ is $\eps$-far from any distribution $(H',Z)$ with $\cp(H',Z) \leq \alpha/(|\Hfam|\cdot M)$. 
\end{enumerate}
\end{lemma}
\begin{proof}
We introduce the following notations first.  For every $h \in \Hfam$, we define $q_h: [M] \rightarrow \R$ by
$$q_h(y) = \min \{ \Pr[ (H,Y) = (h,y)], \Pr[ (H',Z) = (h,y)]\}$$ 
for every $y \in [M]$. We also define $f:\Hfam \rightarrow \R$ by
$$f(h) = \sum_{y \in [M]} q_h(y) \leq \Pr[H=h] = \frac{1}{|\Hfam|}.$$

For the first statement, let $(H',Z)$ be a random variable over $\Hfam \times [M]$ that is $(1-\delta)$-close to $(H,Y)$. We need to show  that $\cp(H',Z) > \alpha/(|\Hfam|\cdot M)$. Note that $\sum_h f(h) = 1-\Delta( (H,Y), (H',Z)) \geq \delta$. So there are at least a $(3\delta/4)$-fraction of hash functions $h$ with $f(h) \geq (\delta/4)/|\Hfam|$. At least a $(3\delta/4)-(\delta/4) = \delta/2$-fraction of $h$ satisfy both $f(h) \geq (\delta/4)/|\Hfam|$ and $Y|_{H=h}$ is $(\delta/4,2\alpha/\delta)$-nonuniform. By the definition of nonuniformity, for each such $h$, we have
$$\sum_{y \in [M]^T} (|\Hfam| \cdot q_h(y))^2 > \frac{2\alpha}{\delta\cdot M}.$$ 
Therefore,
$$\cp(H',Z) \geq \sum_{h,y} q_h(y)^2 > \left(\frac{\delta}{2} \cdot |\Hfam|\right) \cdot \frac{2\alpha}{\delta\cdot
|\Hfam|^2M} = \frac{\alpha}{|\Hfam|\cdot M}.$$ 

Similarly, for the second statement, let $(H',Z)$ be a random variable over $\Hfam \times [M]$ that is $\eps$-close to $(H,Y)$. We need to show  that $\cp(H',Z) > \alpha/(|\Hfam|\cdot M)$. Note that $\sum_h f(h) = 1-\Delta( (H,Y), (H',Z)) \geq 1-\eps$. So there are at least a $1-\eps/0.9$-fraction of $h$ with $f(h) \geq 0.1/|\Hfam|$. At least a $2\eps-\eps/0.9 > \eps/2$-fraction of hash functions satisfy both $f(h) \geq 0.1/|\Hfam|$ and $Y|_{H=h}$ is $(0.1,2\alpha/\eps)$-nonuniform. By Lemma \ref{lem:hard_to_reduce_cp}, for each such $h$, we have
  $$\sum_{y \in [M]} (|\Hfam|\cdot q_h(y))^2 > \frac{2\alpha}{\eps\cdot M}.$$ 
Therefore,
  $$\cp(H',Z) \geq \sum_{h,y} q_h(y)^2 > \left(\frac{\eps}{2} \cdot |\Hfam|\right) \cdot \frac{2\alpha}{\eps\cdot |\Hfam|^2M} = \frac{\alpha}{|\Hfam|\cdot M}.$$ 
\end{proof}
\fi

\subsection{Lower Bounds for the Distribution of Hashed Values Only}

We can extend our lower bounds to the distribution of hashed sequence $\vecY = (H(X_1),\dots,H(X_T))$ along (without $H$) for both closeness requirements, at the price of losing the dependency on $\eps$ and incurring some dependency on the size of the hash family. Let $2^d = |\Hfam|$ be the size of the hash family. The dependency on $d$ is necessary. Intuitively, the hashed sequence $\vecY$ contains at most $T \cdot m$ bits of entropy, and the input $(H,X_1,\dots,X_T)$ contains at least $d + T \cdot k$ bits of entropy. When $d$ is large enough, it is possible that all the randomness of hashed sequence comes from the randomness of the hash family. Indeed, if $H$ is $T$-wise independent (which is possible with $d \simeq T\cdot m$), then $(H(X_1),\dots,H(X_T))$ is uniform when $X_1,\dots,X_T$ are all distinct. Therefore, 
  $$\Delta( (H(X_1),\dots,H(X_T)), U_{[M]^T}) \leq \Pr[ \mbox{ not all $X_1,\dots,X_T$ are distinct }]$$
Thus, $K = \Omega(T^2)$ (independent of $M$) suffices to make the hashed value close to uniform. 

\begin{theorem} \label{thm:lb_no_H}
Let $N,M, T$ be positive integers, and $d$ a positive real number such that $N \geq MT/d$. Let $\delta \in (0,1)$, $\alpha > 1$ be real numbers such that $\alpha \cdot 2^d < \delta^3 \cdot e^{T/32}/128$. Let $H:[N] \rightarrow [M]$ be a random hash function from an hash family $\Hfam$ of size at most $2^d$. There exists an integer $K = \Omega(\delta^2 MT/d\cdot \log(\alpha/\delta))$, and a block $K$-source $\vecX =(X_1,\dots,X_T)$ such that $\vecY = (H(X_1),\dots,H(X_T))$ is $(1-\delta)$-far from any distribution $\vecZ = (Z_1,\dots,Z_T)$ with $\cp(\vecZ) \leq \alpha/M^T$. In particular, $\vecY$ is $(1-\delta)$-far from uniform.
\end{theorem}

Think of $\alpha$ and $\delta$ as constants. Then the theorem says that when the hash function contains $d \leq T/(32\ln 2) - O(1)$ bits of randomness, $K = \Omega(MT/d)$ is necessary for the hashed sequence to be close to uniform. For example, in some typical hash applications, $N = \poly(M)$ and the hash function is $2$-universal or $O(1)$-wise independent. In this case, $d = O(\log M)$ and we need $K = \Omega(MT/\log M)$. (Recall that our upper bound in Theorem \ref{thm:2-univ_cp} says that $K = O(MT)$ suffices.)
\ifnum\short=1
Details are in the full version of this paper\cite{ChungVa08TR}.
\fi

\ifnum\short=0
\begin{proof}
We will deduce the theorem from Theorem \ref{thm:lb_small_cp}. Replacing the parameter $\alpha$ by $\alpha \cdot 2^d$ in Theorem \ref{thm:lb_small_cp}, we know that there exists an integer $K = \Omega(\delta^2 MT/d\cdot \log(\alpha/\delta))$ and  a block $K$-source $\vecX = (X_1,\dots,X_T)$ such that $(H,\vecY) = (H,H(X_1),\dots,H(X_T))$ is $(1-\delta)$-far from any distribution $(H',\vecZ)$ with $\cp(H',\vecZ) \leq \alpha\cdot 2^d/ (2^d\cdot M^T) = \alpha/M^T$. Now, suppose we are given a random variable $\vecZ$ on $[M]^T$ with $\Delta(\vecY,\vecZ) \leq 1-\delta$. Then we can define an $H'$ such that $\Delta((H,\vecY),(H',\vecZ)) = \Delta(\vecY,\vecZ)$ (Indeed, define the conditional distribution $H'|_{\vecZ = \vecz}$ to equal $H|_{\vecY = \vecz}$ for every $\vecz \in [M]^T$.) Then we have 
  $$\cp(\vecZ) \geq \cp(H',\vecZ) > \frac{\alpha}{M^T}.$$
\end{proof}
\fi

\ifnum\short=0
One limitation of the above lower bound is that it only works when $d \leq T/(32\ln 2)- O(1)$. For example, the lower bound cannot be applied when the hash function is $T$-wise independent. Although $d = \Omega(T)$ may not be interesting in practice, for the sake of completeness, we provide another simple lower bound to cover this parameter region. 

\begin{theorem} \label{thm:lb_no_H_large_d}
Let $N,M, T$ be positive integers, and $\delta \in (0,1)$, $\alpha > 1$, $d > 0$ real numbers. Let $H:[N] \rightarrow [M]$ be a random hash function from an hash family $\Hfam$ of size at most $2^d$. Suppose $K \leq N$ be an integer such that $K \leq (\delta^2/4\alpha\cdot 2^d)^{1/T}\cdot M$. Then there exists a block $K$-source $\vecX =(X_1,\dots,X_T)$ such that $\vecY = (H(X_1),\dots,H(X_T))$ is $(1-\delta)$-far from any distribution $\vecZ = (Z_1,\dots,Z_T)$ with $\cp(\vecZ) \leq \alpha/M^T$. In particular, $\vecY$ is $(1-\delta)$-far from uniform.
\end{theorem}

Again, think of $\alpha$ and $\delta$ as constants. The theorem says that $K = \Omega(M/2^{d/T})$ is necessary for the hashed sequence to be close to uniform. In particular, when $d = \Theta(T)$, $K = \Omega(M)$ is necessary. Theorem \ref{thm:lb_no_H} gives the same conclusion, but only works for $d \leq T/(32\ln 2) - O(1)$. On the other hand, when $d = o(T)$, Theorem \ref{thm:lb_no_H} gives better lower bound $K = \Omega(MT/d)$.

\begin{proof}
Let $X$ be any flat $K$-source, i.e., a uniform distribution over a set of size $K$. We simply take $\vecX = (X_1,\dots,X_T)$ to be $T$ independent copies of $X$. Note that $\vecY$ has support at most as large as $(H,\vecX)$. Thus,
  $$ |\supp(\vecY)| \leq |\supp(H,\vecX)| = 2^d \cdot K^T \leq \frac{\delta^2}{4\alpha} \cdot M^T.$$
Therefore, $\vecY$ is $(1-\delta^2/4\alpha)$-far from uniform. By Lemma \ref{lem:hard_to_reduce_cp}, $\vecY$ is $(1-\delta)$-far from any distribution $\vecZ = (Z_1,\dots,Z_T)$ with $\cp(\vecZ) \leq \alpha/M^T$.
\end{proof}
\fi

\subsection{Lower Bound for $2$-universal Hash Functions}

In this subsection, we show Theorem \ref{thm:2-univ_cp} is almost tight in the following sense.  We show that there exists $K = \Omega(MT/\eps \cdot \log(1/\eps) )$, a $2$-universal hash family $\Hfam$, and a block $K$-source $\vecX$ such that $(H,\vecY)$ is $\eps$-far from having collision probability $100/(|\Hfam|\cdot M^T)$. The improvement over Theorem  \ref{thm:lb_small_cp} is the almost tight dependency on $\eps$. Recall that Theorem \ref{thm:2-univ_cp} says that for $2$-universal hash family, $K = O(MT/\eps)$ suffices. The upper and lower bound differs by a factor of $\log (1/\eps)$. In particular, our result for $4$-wise independent hash functions (Theorem \ref{thm:4-wise_cp}) cannot be achieved with $2$-universal hash functions. The lower bound can further be extended to the distribution of hashed sequence $\vecY = (H(X_1),\dots,H(X_T))$ as in the previous subsection. Furthermore, since the $2$-universal hash family we use has small size, we only pay a factor of $O(\log M)$ in the lower bound on $K$. 
\ifnum\short=0
Formally we prove the following results.
\fi
\ifnum\short=1
Formally we have the following result.
\fi


\begin{theorem} \label{thm:lb_2-univ}
For every prime power $M$, real numbers $\eps \in (0,1/4)$ and $\alpha \geq 1$, the following holds. For all integers $t$ and $N$ such that $\eps \cdot M^{t-1} \geq 1$ and $N \geq 6\eps M^{2t}$, and for $T = \lceil \eps^2M^{2t-1} \log (\alpha/\eps) \rceil$,\footnote{For technical reason, our lower bound proof does not work for every sufficiently large $T$. However, note that the density of $T$ such that the lower bound holds is $1/M^2$.} there exists an integer $K = \Omega(MT/\eps\cdot\log(\alpha/\eps))$, and a $2$-universal hash family $\Hfam$ from $[N]$ to $[M]$, and a block $K$-source $\vecX =(X_1,\dots,X_T)$ such that $(H,\vecY) = (H,H(X_1),\dots,H(X_T))$ is $\eps$-far from any distribution $(H',\vecZ)$ with $\cp(H',\vecZ) \leq \alpha/(|\Hfam|\cdot M^T)$.
\end{theorem}

\ifnum\short=0
\begin{theorem} \label{thm:lb_2-univ_no_H}
For every prime power $M$, real numbers $\eps \in (0,1/4)$ and $\alpha \geq 1$, the following holds. For all integers $t$ and $N$ such that $\eps \cdot M^{t-1} \geq 1$ and $N \geq 6\eps M^{2t}$, and for $T = \lceil \eps^2M^{2t-1} \log (\alpha M/\eps) \rceil$, there exists an integer $K = \Omega(MT/\eps\cdot\log(\alpha M/\eps))$, and a $2$-universal hash family $\Hfam$ from $[N]$ to $[M]$, and a block $K$-source $\vecX =(X_1,\dots,X_T)$ such that $\vecY = (H(X_1),\dots,H(X_T))$ is $\eps$-far from any distribution $\vecZ$ with $\cp(\vecZ) \leq \alpha/M^T$.
\end{theorem}
\fi

Basically, the idea is to show that the Markov Inequality applied in the proof of Theorem \ref{thm:2-univ_cp} \ifnum\short=0 (see Inequality (\ref{eq:Markov}))\fi is tight for a single block. More precisely, we show that there exists a $2$-universal hash family $\Hfam$, and a $K$-source $X$ such that with probability $\eps$ over $h \leftarrow H$, $\cp(h(X)) \geq 1/M + \Omega(1/K\eps)$. Intuitively, if we take $T = \Theta(K\eps \cdot \log (\alpha/\eps)/M)$ independent copies of such $X$, then the collision probability will satisfy $\cp(h(X_1),\dots,h(X_T)) \geq (1+\Omega(M/K\eps))^T/M^T \geq \alpha
/(\eps M^T)$, and so the overall collision probability is $\cp(H,\vecY) \geq \alpha /(|\Hfam|\cdot M^T)$. Formally, we analyze our construction below using Hellinger distance, and show that the collision probability remains high even after modifying a $\Theta(\eps)$-fraction of distribution.
\ifnum\short=1
For details and the extension to distribution $\vecY = (H(X_1),\dots,H(X_T))$, please refer to \cite{ChungVa08TR}.
\fi

\ifnum\short=0 

\begin{proofof}{of Theorem \ref{thm:lb_2-univ}} 
Fix a prime power $M$, and $\eps > 0$, we identify $[M]$ with the finite field $\F$ of size $M$. Let $t$ be an integer parameter such that $M^{t-1} > 1/\eps$. Recall that the set $\Hfam_0$ of linear functions $\{h_{\vec{a}}:\F^t\rightarrow \F\}_{\vec{a}\in\F^t}$ where $h_{\vec{a}}(\vec{x}) = \sum_i a_i x_i$ is $2$-universal. Note that
picking a random hash function $h\leftarrow \Hfam_0$ is equivalent to picking a random vector $\vec{a}\leftarrow \F^t$. Two special properties of $\Hfam_0$ are (i) when $\vec{a} = \vec{0}$, the whole domain $\F^t$ is sent to $0 \in \F$, and (ii) the size of hash family $|\Hfam_0|$ is the same as the size of the domain, namely $|\F^t|$. We will use $\Hfam_0$ as a building block in our construction.

We proceed to construct the hash family $\Hfam$. We partition the domain $[N]$ into several sub-domains, and apply different hash function to each sub-domain. Let $s$ be an integer parameter to be determined later. We require $N \geq s\cdot M^t$, and partition $[N]$ into $D_0, D_1,\dots,D_s$, where each of $D_1,\dots,D_s$ has size $M^t$ and is identified with $\F^t$, and $D_0$ is the remaining part of $[N]$. In our construction, the data $\vecX$ will never come from $D_0$. Thus, wlog, we can assume $D_0$ is empty. For every $i = 1,\dots, s$, we use a linear hash function $h_{\vec{a}_i} \in \Hfam_0$ to send $D_i$ to $\F$. Thus, a hash function $h \in \Hfam$ consists of $s$ linear hash function $(h_{\vec{a}_1},\dots,h_{\vec{a}_s})$, and can be described by $s$ vectors $\vec{a}_1,\dots,\vec{a}_s \in \F^t$. Note that to make $\Hfam$ $2$-universal, it suffices to pick $\vec{a}_1,\dots,\vec{a}_s$ pairwise independently. Specifically, we identify $\F^t$ with the finite field $\hat{\F}$ of size $M^t$, and pick $(\vec{a}_1,\dots,\vec{a}_s)$ by picking $a, b \in \hat{\F}$, and output $(a+\alpha_1 \cdot b, a+\alpha_2 \cdot b,\dots, a + \alpha_s \cdot b)$ for some distinct constants $\alpha_1,\dots,\alpha_s \in \hat{\F}$. Formally, we define the hash family to be
  $$\Hfam = \{ h^{a,b} : [N] \rightarrow [M] \}_{a,b\in \hat{\F}} \mbox{, where } h^{a,b} = (h_{a+\alpha_1  b},\dots, h_{a+\alpha_s b}) \eqdef (h^{a,b}_1,\dots,h^{a,b}_s).$$ 
It is easy to verify that $\Hfam$ is indeed $2$-universal, and $|\Hfam| = M^{2t}$. 

We next define a single block $K$-source $X$ that makes the Markov Inequality (\ref{eq:Markov}) tight. We simply take $X$ to be a uniform distribution over $D_1 \cup \cdots \cup D_s$, and so $K = s\cdot M^t$. Consider a hash function $h^{a,b} \in \Hfam$. If all $h^{a,b}_i$ are non-zero and distinct, then $h^{a,b}(X)$ is the uniform distribution. If exactly one $h^{a,b}_i = 0$, then $h^{a,b}$ sends $M^t + (s-1)M^{t-1}$ elements in $[N]$ to $0$, and $(s-1)M^{t-1}$ elements to each nonzero $y\in \F$. Let us call such $h^{a,b}$ \emph{bad} hash functions. Thus, if $h^{a,b}$ is bad, then 
\begin{eqnarray*}
 \cp(h^{a,b}(X)) & = & \left(\frac{M^t + (s-1)M^{t-1}}{K}\right)^2 + (M-1)\cdot \left(\frac{(s-1)M^{t-1}}{K}\right)^2 \\
  & = & \frac{1}{M} + \frac{M-1}{s^2M} \geq \frac{1}{M} + \frac{1}{2s^2}.
\end{eqnarray*}  
Note that $h^{a,b}$ is bad with probability
  $$\Pr[ \mbox{exactly one $h^{a,b}_i = 0$}] = \Pr[ b \neq 0 \wedge \exists i \quad (a + \alpha_i b = 0)] = \left(1-\frac{1}{M^t}\right)\cdot \frac{s}{M^t} \geq \frac{s}{2M^t}.$$ 
We set $s = \lceil 4\eps M^t \rceil \leq M^t$. It follows that with probability at least $2\eps$ over $h\leftarrow \Hfam$, the collision probability satisfies $\cp(h(X)) \geq 1/M + 1/(4K\eps)$, as we intuitively desired. However, instead of working with collision probability directly, we need to use Hellinger closeness to measure the growth of distance to uniform (see Definition \ref{def:Hellinger}.) The following claim upper bounds the Hellinger closeness of $h(X)$ for bad hash functions $h$. The proof of the claim is deferred to the end of this section.

\begin{claim} \label{clm:comp_closeness}
Suppose $h$ is a bad hash function defined as above, then the Hellinger closeness of $h(X)$ satisfies $C(h(X)) \leq 1 -M/(64K\eps)$.
\end{claim}

Finally, for every integer $T \in [\eps^2M^{2t-1} \log (\alpha/\eps), c_0 \cdot \eps^2M^{2t-1} \log (\alpha/\eps)]$, we can write $T = c \cdot (64K\eps/M)\cdot \ln (800\alpha/\eps)$ for some constant $c < c_0$. Let $\vecX = (X_1,\dots,X_T)$ be $T$ independent copies of $X$. We now show that $K, \Hfam, \vecX$ satisfy the conclusion of the theorem.  That is, $K = \Omega( MT/(\eps \log(\alpha/\eps)))$ (as follows from the definition of $T$) and $(H,\vecY) = (H,H(X_1),\dots,H(X_T))$ is $\eps$-far from any distribution $(H',\vecZ)$ with $\cp(H',\vecZ) \leq \alpha/(|\Hfam|\cdot M^T)$.

Consider the distribution $(h(X_1),\dots,h(X_T))$ for a bad hash function $h \in \Hfam$. From the above claim, the Hellinger closeness satisfies
  $$C( h(X_1),\dots,h(X_T)) = C(h(X))^T \leq (1 - M/64K\eps)^T \leq e^{MT/64K\eps} \leq \frac{800\alpha}{\eps}.$$
By Lemma \ref{lem:Hellinger_to_stat} and the definition of Hellinger closeness, we have 
  $$\Delta( (h(X_1),\dots,h(X_T)), U_{[M]^T}) \geq 1 - C( h(X_1),\dots,h(X_T)) \geq 1 - \frac{800\alpha}{\eps}.$$
By Lemma \ref{lem:hard_to_reduce_cp}, $(h(X_1),\dots,h(X_T))$ is $(0.1,2\alpha/\eps)$-nonuniform for at least $2\eps$-fraction of bad hash functions $h$. By the second statement of Lemma \ref{lem:nonuniform_far}, this implies $(H,\vecY)$ is $\eps$-far from any distribution $(H',\vecZ)$ with $\cp(H',\vecZ) \leq \alpha/(|\Hfam|\cdot M^T)$.

In sum, given $M,\eps,\alpha, t$ that satisfies the premise of the theorem, we set $K = \lceil 4\eps M^{t}\rceil \cdot M^t$, and proved that for every $N \geq K$, and $T = \Theta((K\eps/M)\cdot \ln (\alpha/\eps))$, the conclusion of the theorem is true. It remains to prove Claim \ref{clm:comp_closeness}.

\begin{claimproof}
Let $p(x) = M\cdot \Pr[h(X)=x]$ for every $x \in \F$. For a bad hash function $h$, we have $p(0) = (1+ (M-1)/s)$, and $p(x) = (1- 1/s)$ for every $x \neq 0$. We will upper bound $C(h(X)) = (1/M)\cdot \sum_x \sqrt{p(x)}$ using Taylor series. Recall that for $z \geq 0$, there exists some $z',z'' \in [0,z]$ such that
  $$\sqrt{1+z} = 1 + \frac{z}{2} + \frac{z^2}{2} \cdot \left( -\frac{1}{4(1+z')^{3/2}}\right) \leq 1 + \frac{z}{2} - \frac{z^2}{8(1+z)^{3/2}} \mbox{ , and}$$
  $$\sqrt{1-z} = 1 - z \cdot \frac{1}{2\sqrt{1-z''}} \leq 1 - \frac{z}{2}.$$
We have
\begin{eqnarray*}
  C(h(X)) &=& \frac{1}{M}\sum_x \sqrt{p(x)} \\
      &\leq & \frac{1}{M} \left( 1 + \frac{M-1}{2s} - \frac{(M-1)^2}{8s^2\cdot(1+(M-1)/s)^{3/2}} + (M-1)\cdot \left(1 - \frac{1}{2s} \right)  \right) \\   
      & = & 1 - \frac{(M-1)^2}{8Ms^2(1+(M-1)/s)^{3/2}} 
\end{eqnarray*}
Recall that $M\geq 2$, $s = \eps M^t \geq M$, and $s^2 = K\eps$, we have
  $$C(h(X)) \leq 1- \frac{M^2}{64K\eps}.$$   
\end{claimproof}	
\end{proofof}

Recall that $|\Hfam| = M^{2t}$. Theorem \ref{thm:lb_2-univ_no_H} follows from Theorem \ref{thm:lb_2-univ} by exactly the same argument as in the proof of Theorem \ref{thm:lb_no_H}.
\fi

\section*{Acknowledgments}
We thank Wei-Chun Kao for helpful discussions in the early stages of this
work, David Zuckerman for telling us about Hellinger distance, and Michael
Mitzenmacher for suggesting parameter settings useful in practice.

\bibliographystyle{alphabetic}
\bibliography{hash}

\ifnum\short=0

\appendix

\section{Technical Lemma on Binomial Coefficients}
\label{app:sec:Bino}

\begin{lemma}[Claim \ref{clm:binomial}, restated] \label{app:lem:binomial}
Let $N, K > 1$ be integers such that $N > 2K$, and $L\in [0,K/2]$, $\beta \in (0, \min\{1,\sqrt{L}\})$ real numbers. Let $S \subset [N]$ be a random subset of size $K$, and $T \subset [N]$ be a fixed subset of $[N]$ of arbitrary size. We have
$$ \Pr_S \left[ \left| |S\cap T| - L\right| \leq \beta \sqrt{L} \right] \leq O(\beta).$$
\end{lemma}
\begin{proof}
By an abuse of notation, we use $T$ to denote the size of set $T$. The probability can be expressed as a sum of binomial coefficients as follows.
$$\Pr_S \left[ \left| |S\cap T| - L\right| \leq \beta \sqrt{L} \right] = \sum_{R=\left\lceil L - \beta \sqrt{L}\right\rceil}^{\left\lfloor L + \beta \sqrt{L}\right\rfloor } \frac{{T\choose R}{N-T\choose K-R}}{{N\choose K}}.$$
Note that there are at most $\lfloor 2\beta \sqrt{L}\rfloor + 1$ terms, it suffices to show that for every $R \in \left[L - \beta \sqrt{L},L + \beta \sqrt{L}\right]$,
$$f(T) \eqdef \frac{{T\choose R}{N-T\choose K-R}}{{N\choose K}} \leq O\left(\sqrt{\frac{1}{L}}\right).$$ 
We use the following bound on binomial coefficients, which can be derived from Stirling's formula.
\begin{claim}
  For integers $0< i < a$, $0<j<b$, we have 
  $$\frac{{a\choose i}{b\choose j}}{{a+b\choose i+j}} \leq O\left(\sqrt{\frac{a\cdot b \cdot (i+j) \cdot (a+b-i-j)}{i\cdot(a-i) \cdot j \cdot (b-j) \cdot (a+b)}}\right).$$
\end{claim}
Note that $L\in[0,K/2]$ implies $K-R = \Omega(K)$. When $2R \leq T \leq N-2K + 2R$, we have
\begin{eqnarray*}
\lefteqn{f(T) = \frac{{T\choose R}{N-T\choose K-R}}{{N\choose K}}} \\
& = & O\left( \sqrt{ \frac  
            {T (N-T) K (N-K)} {R(T-R)(K-R)(N-T-K+R)N}
            }\right) \\
& = & O\left( \sqrt{ 
            \frac{1}{R} \cdot
            \frac{K}{K-R} \cdot
            \frac{N-K}{N} \cdot
            \frac{T(N-T)}{(T-R)(N-T-K+R)}
            }\right) \\
& = & O\left(\sqrt{\frac{1}{R}}\right) = O\left(\sqrt{\frac{1}{L}}\right),
\end{eqnarray*}
as desired. Note that when $N > 2K$, such $T$ exists. Finally, observe that $\beta^2 < L$ implies $R \geq 1$, and
$$ \frac{f(T)}{f(T+1)} = \frac{(T-R+1)(N-T)}{(T+1)(N-T-K+R)}.$$
It follows that $f(T)$ is increasing when $T \leq 2R$, and $f(T)$ is decreasing when $T \geq N-2K+2R$. Therefore, $f(T) \leq f(2R) = O(\sqrt{1/L})$ for $T \leq 2R$, and $f(T) \leq f(N-2K+2R) = O(\sqrt{1/L})$ for $T \geq N-2K+2R$, which complete the proof. 
\end{proof}

\section{Proof of Lemma \ref{lem:stat_inc_small_T}} \label{app:proof_lem_stat_inc_small_T}

\begin{lemma}[Lemma \ref{lem:stat_inc_small_T}, restated] \label{app:lem:stat_inc_small_T}
  Let $X$ and $Y$ be random variables over $[M]$ such that
  $\Delta(X,Y) \geq \eps$. Let $\vecX = (X_1,\dots,X_T)$ be $T$
  i.i.d. copies of $X$, and let $\vecY = (Y_1,\dots,Y_T)$ be $T$
  i.i.d. copies of $Y$. We have
  $$ \Delta(\vecX,\vecY) \geq \min\{ \eps_0, c\sqrt{T}\cdot\eps\},$$
  where $\eps_0, c$ are absolute constants.
\end{lemma}

\begin{proof}
We prove the lemma by the following two claims. The first claim
reduces the lemma to the special case that $X$ is a Bernoulli random
variable with bias $\Omega(\eps)$, and $Y$ is a uniform coin. The
second claim proves the special case.

\begin{claim} \label{lem:reduce_Bernoulli}
  Let $X$ and $Y$ be random variables over $[M]$ such that
  $\Delta(X,Y) = \eps$. Then there exists a randomized
  function $f:[M] \rightarrow \zo$ such that $f(Y) = U_{\zo}$,
  and $\Delta(f(X), f(Y)) \geq \eps/2$.
\end{claim}
\begin{claimproof}
By the definition, there exists a set $T \subset [M]$ such that
$$\left| \Pr[X \in T] - \Pr[Y \in T] \right| = \eps.$$
With out loss of generality, we can assume that $\Pr[Y\in T] = p
\leq 1/2$ (because we can take the complement of $T$.) Let $g : [M]
\rightarrow \zo$ be the indicator function of $T$, so we have
$\Pr_Y[g(Y)=1] = p$. For every $x\in[M]$, we define $f(x) = g(x)
\vee b$, where $b$ is a biased coin with $\Pr[b = 0] = 1/(2(1-p))$.
The claim follows by observing that  $$\Pr[f(Y) = 0] = \Pr[g(Y) = 0
\wedge b = 0] = (1-p) \cdot 1/(2(1-p)) = 1/2,$$ and
$$\Delta(f(X), f(Y)) \geq \Delta(X,Y) \cdot \Pr[b=0] \geq \eps/2.$$
\end{claimproof}
\begin{claim} \label{lem:indep_Bernoulli}
  Let $X$ be a Bernoulli random variable over $\zo$ such that
  $\Pr[X=0] = 1/2 - \eps$. Let $\vecX = (X_1,\dots,X_T)$ be $T$
  independent copies of $X$. Then
  $$ \Delta(\vecX, U_{\zo^T}) \geq \min \{ \eps_0, c \sqrt{T} \eps \},$$
  where $\eps_0, c$ are absolute constants independent of $\eps$ and $T$.
\end{claim}
\begin{claimproof}
For $\vecx \in \zo^T$, let the \emph{weight} $\wt(\vecx)$ of $\vecx$ to be
the number of $1$'s in $\vecx$. Let
$$S = \left\{x\in \zo^T: \wt(x) \leq \frac{T}{2} -\sqrt{T} \right\}$$
be the subset of $\zo^T$ with small weight (This choice of $S$ is the
main source of improvement in our proof compared to that of Reyzin \cite{Reyzin04},
who instead considers the set of all $x$ with weight at most $T/2$.) For every $\vecx \in S$,
we have
$$\Pr[ \vecX = \vecx] \leq \frac{1}{2^T} \cdot (1-\eps)^{T/2 + \sqrt{T}}
\cdot (1+\eps)^{T/2-\sqrt{T}} \leq
 \left(1- \min \left\{ \frac{\sqrt{T}\cdot
\eps}{2} ,\frac{1}{2}\right\} \right) \cdot \Pr[ U_{\zo^T} = \vecx]
.$$ 
By standard results on large deviation, we have 
$$\Pr[U_{\zo^T} \in S] \geq \Omega(1).$$
Combining the above two inequalities, we get
\begin{eqnarray*}
  \Delta(\vecX, U_{\zo^T}) & \geq & \Pr[U_{\zo^T} \in S] - \Pr[\vecX \in S]  
\\
  & \geq & \left(1 - \left(1- \min \left\{ \frac{\sqrt{T}\cdot
\eps}{2} ,\frac{1}{2}\right\} \right) \right) \cdot  \Pr[U_{\zo^T}
\in S] \\
  & \geq & \min \left\{ \frac{\sqrt{T}\cdot
\eps}{2} ,\frac{1}{2}\right\} \cdot \Omega(1) \\
  & = & \min \{ c \sqrt{T} \eps, \eps_0 \}
\end{eqnarray*}
for some absolute constants $c, \eps_0$, which completes the proof.
\end{claimproof}
Note that applying the same randomized function $f$ on two random
variables $X$ and $Y$ cannot increase the statistical distance.
I.e., $\Delta(f(X),f(Y)) \leq \Delta(X,Y)$. The lemma following
immediately by the above two claims:
\begin{eqnarray*}
  \Delta(\vecX, \vecY) & \geq & \Delta( ((f_1(X_1),\dots,f_T(X_T)),
   ((f_1(Y_1),\dots,f_T(Y_T)) \\
   & \geq & \min \{ \eps_0, c \sqrt{T} \eps \}
\end{eqnarray*}
where $f_1,\dots,f_T$ are independent copies of randomized function
defined in Claim \ref{lem:reduce_Bernoulli}, and $\eps_0,c$ are
absolute constants from Claim \ref{lem:indep_Bernoulli}.
\end{proof}

\fi

\end{document}

%% file: macros.tex
\ifnum\short=0
\newtheorem{theorem}{Theorem}[section]
\newtheorem{definition}[theorem]{Definition}
\newtheorem{lemma}[theorem]{Lemma}

\newtheorem{claim}[theorem]{Claim}

\newtheorem{remk}[theorem]{Remark}

\newenvironment{proof}{\noindent{\bf Proof. }}{\qed}
\fi

\def\FullBox{\hbox{\vrule width 8pt height 8pt depth 0pt}}

\def\qed{\ifmmode\qquad\FullBox\else{\unskip\nobreak\hfil
\penalty50\hskip1em\null\nobreak\hfil\FullBox
\parfillskip=0pt\finalhyphendemerits=0\endgraf}\fi}

\def\qedsketch{\ifmmode\Box\else{\unskip\nobreak\hfil
\penalty50\hskip1em\null\nobreak\hfil$\Box$
\parfillskip=0pt\finalhyphendemerits=0\endgraf}\fi}

\newenvironment{proofof}[1]{\begin{trivlist} \item {\bf Proof #1:~~}}
  {\qed\end{trivlist}}
\newenvironment{claimproof}{\begin{quotation} \noindent
{\bf Proof of claim:~~}}{\qedsketch\end{quotation}}

\newcommand{\eqdef}{\mathbin{\stackrel{\rm def}{=}}}
\newcommand{\R}{{\mathbb R}}

\newcommand{\F}{\mathbb{F}}

\newcommand{\poly}{{\mathrm{poly}}}

\newcommand{\zo}{\{0,1\}}

\newcommand{\getsr}{\gets}
\newcommand{\E}{\mathop{\mathrm E}\displaylimits}
\newcommand{\Var}{\mathop{\mathrm{Var}}\displaylimits}

\newcommand{\remove}[1]{}
\newcommand{\Exp}{\mathop{\mbox{\sc E}}\nolimits}

\newcommand{\eps}{\varepsilon}

